\documentclass[10pt,journal]{IEEEtran}
\usepackage{multicol}
\usepackage{amsmath,epsfig,comment}
\usepackage{amsthm}
\usepackage{tcolorbox}
\usepackage{amssymb}
\usepackage{multirow}
\usepackage{mathrsfs}
\usepackage{amsmath}
\usepackage{arydshln}
\usepackage{multirow}
\usepackage{arydshln}

\usepackage{cases} 
\usepackage{amssymb,amsmath,cite}
\usepackage{epsfig}
\usepackage{color}
\usepackage{bm}
\usepackage{graphicx,subfigure}
\usepackage{algorithm}
\usepackage{algorithmic}
\usepackage{multirow}
\usepackage{soul}
\usepackage{geometry}
\usepackage{extarrows}

\def\bd{\mathbf{d}}
\def\bX{\mathbf{X}}
\def\bB{\mathbf{B}}

\def\d{\mathbf{d}}

\def\R{\mathbb{R}}

\def\C{\mathbb{C}}

\def\bW{\mathbf{W}}
\def\bV{\mathbf{V}}
\def\by{\mathbf{y}}

\def\bH{\mathbf{H}}
\def\bU{\mathbf{U}}
\def\bQ{\mathbf{Q}}

\def\bP{\mathbf{P}}

\def\y{\mathbf{y}}
\def\Y{\mathbf{Y}}
\def\x{\mathbf{x}}

\def\h{\mathbf{h}}

\def\bA{\mathbf{A}}

\def\bD{\mathbf{D}}

\def\bI{\mathbf{I}}
\def\bH{\mathbf{H}}
\def\bB{\mathbf{B}}

\def\bh{\mathbf{h}}
\def\bthe{\boldsymbol{\Theta}}

\def\bw{\mathbf{w}}
\def\bd{\mathbf{d}}

\def\v{\mathbf{v}}
\def\bw{\mathbf{w}}
\def\bY{\mathbf{Y}}

\newtheorem{theorem}{Theorem}
\newtheorem{example}{Example}
\newtheorem{lemma}{Lemma}
\newtheorem{remark}{Remark}

\newtheorem{definition}{Definition}
\newtheorem{proposition}{Proposition}

\newtheoremstyle{noparens}%
  {}{}%
  {\itshape}{}%
  {\bfseries}{.}%
  { }%
  {\thmname{#1}\thmnumber{ #2}\mdseries\thmnote{ #3}}

\theoremstyle{noparens}

\title{Possible Research Questions}
\date{today}
\title{RIS Optimization and Scaling Laws in Multi-Operator Systems: Is Quadratic Scaling Achievable?}
\author{\IEEEauthorblockN{Zheyu Wu, Matteo Nerini, and Bruno Clerckx}
  	\thanks{Z. Wu, M. Nerini,  and B. Clerckx are with the Department of Electrical and Electronic Engineering, Imperial College London, London, SW7 2AZ, U.K. (email: \{zheyu.wu, m.nerini20, b.clerckx\}@imperial.ac.uk). }
  }
\date{today}
\begin{document}
\maketitle
\begin{abstract}
This paper studies multi-operator wireless communication systems aided by general reconfigurable intelligent surface (RIS), including both conventional single-connected RIS (also known as diagonal RIS) and beyond-diagonal RIS (BD-RIS). Specifically, we consider a system where multiple operators coexist in the same area over different frequency bands, each with a single-antenna base station, while one operator serves its single-antenna user with the aid of an RIS. In such a system, the RIS may unintentionally reflect signals from the non-serving operators, leading to inter-operator interference and rapid fluctuations of their effective channels.  To address this issue, we propose a practical RIS design framework that maximizes the received signal power of the serving operator while enforcing fixed RIS-reflected channels of the non-serving operators. We derive closed-form solutions to the resulting optimization problem, based on a novel  technique to deal with the coupled  unitary and linear equality constraints.  We further give scaling law analysis of the received signal power. 
For a two-operator system, the received signal power scales quadratically with the number of RIS elements for group-connected BD-RIS with group size $G_s\geq 2$, whereas for conventional single-connected RIS it scales only linearly. More generally, for an $L$-operator system with $L-1$ non-serving operators, the scaling-law transition occurs at  $G_s=L$, where quadratic scaling is achieved when $G_s\geq L$, and linear scaling otherwise. These results demonstrate that, in a multi-operator system, quadratic scaling is achievable only with BD-RIS architectures having enough interconnections. Simulation results validate the analysis and show the significant gain of BD-RIS over conventional RIS in multi-operator systems. In particular, group-connected BD-RIS with $G_s=2$ achieves a $13$ dB gain over conventional RIS in a two-operator system with a 128-element RIS.
 
 \end{abstract}
\title{RIS Optimization in Multi-Operator Systems: Beyond-Diagonal Architectures Are Essential}
\section{Introduction}
Reconfigurable intelligent surface (RIS) is regarded as a promising technology for future wireless communication systems. By tuning a large number of passive reflecting elements, RIS can enhance the received signal power with very low hardware cost and power consumption. In particular,  in a RIS-aided wireless communication system, the received signal power scales quadratically with the number of RIS elements \cite{RIS1}. 
Beyond that, RIS can effectively extend coverage, suppress interference, and improve communication reliability. Owing to these appealing features, RIS has attracted significant research attention in recent years \cite{RIS1,RIS_tutorial,RIS_survey,Pan_overview_2022}.

Conventional RIS, also known as diagonal RIS (D-RIS),  typically employ a single-connected structure, where each element is controlled independently by a reconfigurable impedance. Recently, the concept of RIS has been greatly generalized to include also beyond-diagonal architectures that allow inter-connection between RIS elements.  In particular, RIS elements can be interconnected either all together or in groups, leading to fully-connected and group-connected RIS architectures, respectively \cite{BDRIS,group_conn}.  These beyond-diagonal RIS (BD-RIS) architectures offer enhanced flexibility in controlling the wireless channel and have been shown to achieve superior performance compared with conventional single-connected RIS\cite{BDRIS,group_conn,closeform,tutorialbdris}.  

Extensive research efforts have been devoted to the modeling \cite{Zparameter,generalmodel,lossy_peng}, signal processing \cite{Pan_overview_2022,RIS_sumrate,channel,wu}, and performance analysis \cite{Badiu_analysis_2020,Zhang_analysis_2024,mutualcoupling2} of both conventional single-connected RIS and BD-RIS. 
However, most existing studies and their results, including the favorable quadratic scaling property, are limited to single-operator scenarios, where only one operator controls the RIS to serve its users.
This assumption is overly simplistic. In realistic wireless deployments, multiple operators often coexist within the same geographical area, operate over non-overlapping frequency bands, and share the same propagation environment. Although the RIS is controlled by only one operator to enhance its own network performance, its passive elements reflect all incident signals indiscriminately. This introduces two major practical challenges for the operators that do not control the RIS, hereafter referred to as the non-serving operators. First,  the RIS creates unintended reflections toward the non-serving operators, which may degrade their performance. This effect is referred to as inter-operator interference (IOI). 
Second, the effective channels of non-serving operators vary with the RIS configuration and are difficult to estimate and track, making it highly challenging for them to reliably serve their own users.

These issues have been recognized and analyzed in several recent studies on conventional RIS \cite{Gurgunoglu_estimation_2024, Gurgunoglu_estimation_2026, Yashvanth_2024_single,Yashvanth_2026_single_2,Ghosh_2025_multi,Miridakis_2024_multi,Miridakis_2026_multi_2,Cai_optimize_2022,Lin_optimize_2025,Zhang_optimize_2025}.  The works in \cite{Gurgunoglu_estimation_2024} and \cite{Gurgunoglu_estimation_2026} investigated channel estimation in multi-operator RIS-assisted systems, analyzing  the impact of inter-operator pilot contamination and proposing mitigation schemes.  For downlink transmission,  \cite{Yashvanth_2024_single} and \cite{Yashvanth_2026_single_2} showed that, in a two-operator system where only one operator controls one or multiple RISs, the RIS unintentionally enhances the channel of the other operator as well. The works in \cite{Ghosh_2025_multi,Miridakis_2024_multi,Miridakis_2026_multi_2} further considered a multi-RIS multi-operator system, where each operator deploys its own RIS.    For mmWave channels, it has been demonstrated in \cite{Ghosh_2025_multi} that the performance of a given operator is not degraded by RISs controlled by non-serving operators, due to the spatial sparsity of mmWave channels. In contrast, \cite{Miridakis_2024_multi} and \cite{Miridakis_2026_multi_2} reached the opposite conclusion by assuming Rayleigh/Rician and correlated Nakagami-$m$ fading channels, respectively. They proved that the receive signal-to-noise ratio (SNR) gain enabled by RIS degrades from quadratic scaling with respective to the number of RIS elements in the absence of IOI to linear scaling in its presence.   Several works have also investigated the joint optimization of RISs across different operators  to enhance overall system performance \cite{Cai_optimize_2022,Lin_optimize_2025,Zhang_optimize_2025}. However, these approaches typically rely on strong assumptions, such as shared channel state information (CSI) among operators or the presence of a third-party controller to coordinate RISs across operators.  In practice, such assumptions are often difficult to satisfy as they incur significant challenges including heavy signaling overhead and privacy concerns.


To date, the practical challenges and  performance limits  of multi-operator RIS-aided systems remain largely unresolved  and not yet fully understood.  First, existing downlink performance  analyses and conclusions  rely heavily on the specific channel models and restrictive system assumptions, such as whether only one operator is aided by RIS or each operator deploys its own RIS \cite{Yashvanth_2024_single,Yashvanth_2026_single_2,Ghosh_2025_multi,Miridakis_2024_multi,Miridakis_2026_multi_2}. 
Second, existing analyses mainly focus on characterizing the impact of IOI on received signal power or achievable rate, under the assumption that each operator only serves a single-antenna user \cite{Yashvanth_2024_single,Yashvanth_2026_single_2,Ghosh_2025_multi,Miridakis_2024_multi,Miridakis_2026_multi_2}. In contrast, the issue of rapid fluctuations of the effective channels experienced by non-serving operators, and the resulting CSI acquisition challenges in multi-user settings, have received much less attention and remain unresolved.    Third,  most existing IOI mitigation approaches require cross-operator coordination \cite{Cai_optimize_2022,Lin_optimize_2025,Zhang_optimize_2025}, which is difficult to implement and often impractical. Fourth, the performance of more general BD-RIS architectures in multi-operator scenarios has not been investigated. It remains an open question whether BD-RIS, with its greater design flexibility, can mitigate the challenges of multi-operator systems  and achieve  strong system performance.

In this paper, we present the first study of multi-operator systems assisted by general RIS, including both conventional single-connected RIS and BD-RIS. We consider a multi-operator setting in which each operator is associated with a single-antenna base station (BS). Among them, one operator, referred to as the serving operator, serves its single-antenna user with the aid of an RIS,  while the system configurations of the non-serving operators are left general.  Based on this setup, we propose a practical RIS optimization framework that controls IOI and prevents effective-channel fluctuations of non-serving operators.  We further theoretically characterize the performance of different RIS architectures  under the proposed framework. In particular, we address the following critical question: in the considered multi-operator setting, which RIS architectures can preserve the favorable quadratic scaling of the single-operator scenario? 
 The main contributions are summarized as follows. 

First, we introduce a practical model that explicitly controls the impact of RIS on non-serving operators.  Specifically, the RIS is optimized for its serving operator while enforcing fixed reflected channels of the non-serving operators. The proposed model offers three key advantages: (i) it prevents the RIS from continuously changing the channels experienced by the non-serving operators, thereby  eliminating IOI and avoiding  rapid channel fluctuations; (ii) The network configurations of the non-serving operators, such as whether they are assisted by RISs and how many users they serve, are kept fully general, which encompasses existing models in \cite{Yashvanth_2024_single,Yashvanth_2026_single_2,Ghosh_2025_multi,Miridakis_2024_multi,Miridakis_2026_multi_2} as special cases; (iii) The proposed mechanism only requires exchanging the CSI of the quasi-static BS-RIS channels of the non-serving operators, which greatly reduces the signaling overhead compared to  \cite{Cai_optimize_2022,Lin_optimize_2025,Zhang_optimize_2025}.

Second, we derive closed-form globally optimal solutions to the resulting optimization problem for fully-, group-, and single-connected RIS. In particular, to handle the coupled unitary and linear equality constraints induced by the BD-RIS structure and the fixed reflected-channel constraints, we introduce a novel technique that reformulates these constraints into an equivalent form expressed in terms of a lower-dimensional unconstrained unitary matrix. Based on this transformation, the original problem reduces to a lower-dimensional optimization problem involving only unitary constraints and can be solved in closed form.

Third, we theoretically characterize the received signal power achieved by different RIS architectures under the proposed scheme. For a two-operator system, we show that group-connected BD-RIS with group size $G_s\geq 2$ achieves a received signal power scaling quadratically with the number of RIS elements, while conventional single-connected RIS achieves only linear scaling, under Rayleigh fading RIS-user channels and Rayleigh/line-of-sight (LoS) BS-RIS channels. More generally, for an $L$-operator system with $L-1$ non-serving operators, we prove that the scaling-law transition  occurs at group size $G_s=L$:  quadratic scaling is achieved when $G_s\geq L$ whereas only linear scaling is achievable when $G_s<L$, under Rayleigh fading for both RIS-user and BS-RIS channels. Numerical results further show that these scaling behaviors, although established under specific channel assumptions, remain valid for Rayleigh, Rician, and LoS BS-RIS channels.

Fourth, we discuss the generality of the proposed framework and its analytical insights. In particular, for more general multi-antenna multi-user systems, the proposed technique for handling coupled unitary and linear equality constraints can still be applied to solve the corresponding utility-maximization (e.g.,  received signal power maximization or sum-rate maximization) problems. 
 Numerical results suggest that, in the multi-antenna case, the scaling of the received signal power transitions at group size $G_s=N_T(L-1)+1$,  where $N_T$ is the number of transmit antennas at each BS: the scaling is quadratic  when $G_s\geq N_T(L-1)+1$ and linear otherwise. This requirement on $G_s$ can be further relaxed when the channels from the BSs of the non-serving operators to the RIS are rank-deficient.


\emph{Organization:}
Section \ref{sec:system} introduces the RIS-aided multi-operator system model and formulates the problem. Sections \ref{sec:solution} and \ref{sec:scalinglaw} focus on the two-operator case, deriving closed-form solutions to the formulated problem and providing scaling-law analysis, respectively. These results are then extended to more general multi-operator systems in Section \ref{sec:multioperator}. Section \ref{subsec:multiantenna} discusses the generalization to multi-antenna multi-user case.  Simulation results are provided in Section \ref{sec:simulation}, and Section \ref{sec:conclusion} concludes the paper. 

\emph{Notations}: Throughout the paper, we use italic, bold lowercase, and bold uppercase letters to denote scalar, column vector, and matrix, respectively.  For a matrix $\mathbf{X}$,  $[\bX]_{i_1:i_2,j_1:j_2}$ denotes the submatrix formed by rows $i_1$ to $i_2$ and columns $j_1$ to $j_2$. The operators  $(\cdot)^T$, $(\cdot)^H$, $(\cdot)^*$, and $(\cdot)^{-1}$  return the transpose,  the Hermitian transpose, the conjugate, and the inverse of their arguments, respectively. The notations $|\cdot|$ and $\|\cdot\|$ denote the magnitude of a complex scalar and the $\ell_2$ norm  of a vector, respectively.  For a positive semidefinite matrix $\bX$, $\bX^{\frac{1}{2}}$ refers to its square root. The symbol $\mathbf{I}_N$ denotes the $N\times N$ identity matrix.  The notation $\mathcal{CN}(\mathbf{0},\mathbf{I}_N)$ denotes the circularly symmetric complex Gaussian distribution with mean zero and covariance matrix $\mathbf{I}_N$, and $\chi^2_N$ denotes the chi-square distribution with $N$ degrees of freedom.   The notation  $\Gamma(x)$ denotes the Gamma function, and $\mathcal{O}(\cdot)$ denotes the standard big-$\mathcal{O}$ notation. For two random variables $X$ and $Y$, $X\mid Y$ denotes the random variable $X$ conditioned on $Y$. 
For two sequences $x_N$ and $y_N$, the notation $x_N \overset{N\to\infty}\sim  y_N$ means that $\lim_{N\to\infty}\frac{x_N}{y_N}=1$.  

\section{System Model and Problem Formulation}\label{sec:system}
Consider a wireless communication system in which $L$ independent operators serve their respective users within the same geographical area. Their respective transmitters are assumed to be fixed-location BSs and are denoted by BS$_1$, BS$_2$, $\dots$, BS$_L$. The operators employ non-overlapping frequency bands, and therefore no out-of-band interference exists. The transmission of operator 1 is assisted by an $N$-element RIS; see Fig. \ref{sysmodel}. Accordingly, operator 1, which controls the RIS, is referred to as the \emph{serving operator}, whereas the other $L-1$ operators, which do not control the RIS, are referred to as the \emph{non-serving operators}. 
\begin{figure}
\includegraphics[width=0.43\textwidth]{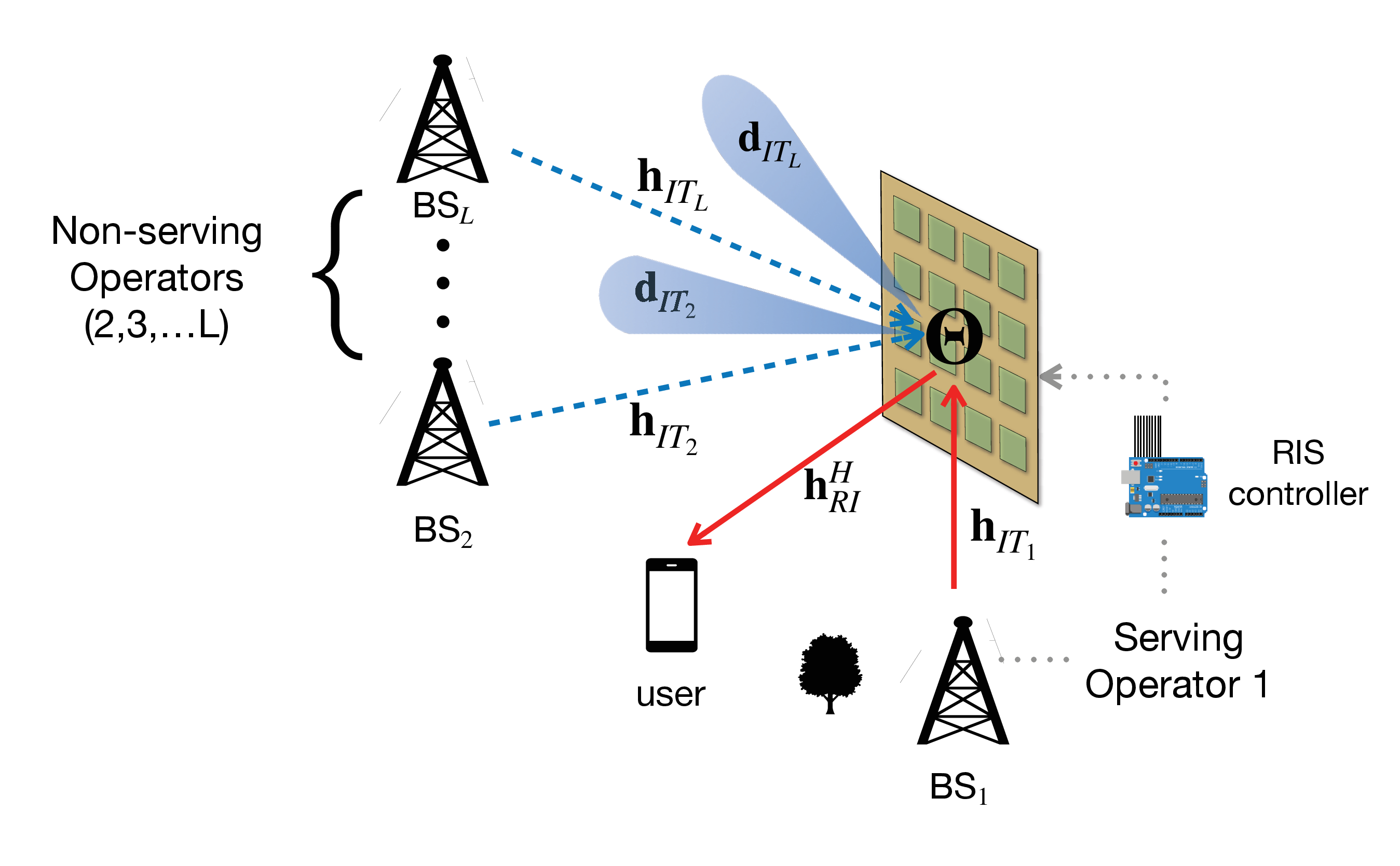}
\centering
\caption{Illustration of a multi-operator system model, where operator 1 is assisted by an RIS. The RIS-reflected channel of BS$_l$ is fixed as $\mathbf{d}_{IT_l}\in\C^{N}$ for $l=2,3,\dots,L$.  The users and infrastructures associated with the non-serving operators are not specified and are therefore omitted from the figure.}
\label{sysmodel}
\end{figure}

For simplicity, we assume that operator 1 serves a single-antenna user and that the BSs associated with all operators are equipped with single antennas\footnote{No other modeling constraints are imposed on the non-serving operators. For example, they may serve one or multiple users and may be aided by other RISs or not. Therefore, the system model in Fig. \ref{sysmodel} includes existing models in \cite{Yashvanth_2024_single,Yashvanth_2026_single_2,Ghosh_2025_multi,Miridakis_2024_multi,Miridakis_2026_multi_2} as special cases.
}. Let $\h_{IT_l}\in\C^{N\times 1}$ denote the BS$_l$-RIS channel, where $l=1,2,\dots, L$, and let  $\h_{RI}^H\in\C^{1 \times N}$ denote the RIS-user channel of operator 1. To focus on the effect of the RIS, we assume that the direct path between BS$_1$ and the user is blocked. The proposed approach and analysis can be readily extended to the case where the direct path is included; see the discussions at the end of Section \ref{subsec:group}. The effective channel from   BS$_1$  to the user then reads $$h=\h_{RI}^H\bthe\h_{IT_1},$$ where  $\bthe\in\C^{N\times N}$ is the scattering matrix of the RIS. The received signal at the user is $y=hx+n,$ where $x\in\C$ is the transmit signal, and $n\sim\mathcal{CN}(0,\sigma^2)$ denotes the additive white Gaussian noise. We further assume that the frequency separations among the operators are sufficiently small such that the frequency-dependent variation of the RIS response can be neglected. Under this assumption, the same RIS scattering matrix $\bthe$ applies to all operators.

In conventional single-operator RIS-aided designs, the scattering matrix $\bthe$ is optimized to maximize the received signal power, which is given by $P_R=P_T|h|^2$, where $P_T$ is the transmitted power \cite{BDRIS,closeform}. However, this approach is problematic in a multi-operator scenario. Since the non-serving operators do not control the RIS,  rapid updates of $\bthe$ performed by operator 1 lead to time-varying and unpredictable effective channels for the other operators, making it difficult for them to serve their users reliably.
  
To address this issue, we maximize $P_R$ for operator 1 while constraining the RIS-reflected channels of the non-serving operators, i.e., $\bthe \h_{IT_l}$ for $l=2,\dots,L$, to be fixed to vectors $\mathbf{d}_{IT_l} \in \mathbb{C}^{N\times 1}$ (see Fig. \ref{sysmodel}), which can be chosen arbitrarily as long as they satisfy a feasibility condition that will be later specified. 
 Since the BSs-RIS channels $\h_{IT_l},~l=1,2,\dots,L,$ are quasi-static and vary much more slowly than the RIS-user channel $\h_{RI}$, constraining $\bthe\h_{IT_l}=\mathbf{d}_{IT_l}$, $l=2,\dots, L$, fixes the impact of the RIS on the non-serving operators over a relatively long time scale, while still allowing $\bthe$ to adapt to $\h_{RI}$ to improve the performance of operator 1. This design requires only the exchange of $\h_{IT_l}$ between operator $l$ and operator 1, where $l=2,\dots,L$, resulting in infrequent information exchange and  reduced signaling overhead compared to \cite{Cai_optimize_2022,Lin_optimize_2025,Zhang_optimize_2025}.

Regarding the RIS architectures, we focus on the general group-connected RIS in this paper. For a lossless group-connected RIS consisting of $G$ groups with equal size, the scattering matrix exhibits a block  diagonal structure as  $\bthe=\text{diag}(\bthe_1,\bthe_2,\dots,\bthe_G)$, where $\bthe_g^H\bthe_g=\mathbf{I}_{G_s}$ for all $g=1,2,\dots, G$, and $G_s=N/G$ is the group size\footnote{This paper focuses on non-reciprocal BD-RIS \cite{Li2025nonreciprocal,Liu2026nonreciprocal} to study the fundamental limits of BD-RIS in multi-operator systems. The analysis of reciprocal BD-RIS, with additional symmetry constraint on $\bthe_g$ in problem \eqref{problem00}, is left for future work.}. In particular, fully-connected RIS and conventional single-connected RIS are special cases of group-connected RIS with $G=1$ and $G=N$, respectively. 

Based on the above discussions, the RIS scattering matrix design for the considered multi-operator systems can be cast as the following optimization problems: 
\begin{subequations}\label{problem00}
\begin{align}
\max_{\bthe}~&|\h_{RI}^H\bthe\h_{IT_1}|^2\label{PR}\\
\text{s.t. }~&\bthe\h_{IT_l}=\mathbf{d}_{IT_l},~l=2,3,\dots,L,\label{con01}\\
&\bthe=\text{diag}(\bthe_1,\bthe_2,\dots,\bthe_G),\label{con02}\\
&\bthe_{g}^H\bthe_{g}=\mathbf{I}_{G_s},~g=1,2,\dots, G,\label{con03}
\end{align}
\end{subequations}
where $P_T$ is set to $1$ without loss of generality.
\begin{remark}
In this paper, we treat $\{\d_{IT_l}\}_{2\leq l\leq L}$ as fixed vectors. As will be shown in Theorems \ref{scalinglaw}--\ref{scalinglaw3}, when the RIS-user channel is Rayleigh fading, the expected received signal power is independent of $\{\d_{IT_l}\}_{2\leq l\leq L}$. For other channel distributions, however, the choice of $\{\d_{IT_l}\}_{2\leq l\leq L}$ may affect the expected received signal power and can be further optimized based on $\{\h_{IT_l}\}_{2\leq l\leq L}$ and the distribution of $\h_{RI}$. Investigation of the impact of $\{\d_{IT_l}\}_{2\leq l\leq L}$ and its design under other channel distributions is left for future work.
\end{remark}
In Sections \ref{sec:solution} and \ref{sec:scalinglaw}, we focus on  a two-operator system  to introduce the main ideas for solving problem \eqref{problem00} and to gain insights into how the scaling law  changes compared with the single-operator case. We then extend the analysis to the general multi-operator setting in Section \ref{sec:multioperator}. A brief discussion on the more general multi-antenna multi-user case is provided in Section \ref{subsec:multiantenna}.
\section{RIS Optimization for the Two-Operator Model}\label{sec:solution}
In this section, we focus on a two-operator system, i.e., $L=2$. The optimization problem in \eqref{problem00} reduces to 
\begin{subequations}\label{problem0}
\begin{align}
\max_{\bthe}~&|\h_{RI}^H\bthe\h_{IT_1}|^2\\
\text{s.t. }~&\bthe\h_{IT_2}=\mathbf{d},\label{problem0:con1}\\
&\bthe=\text{diag}(\bthe_1,\bthe_2,\dots,\bthe_G),\\
&\bthe_{g}^H\bthe_{g}=\mathbf{I}_{G_s},~g=1,2,\dots, G.\label{problem0:con2}
\end{align}
\end{subequations}
In \eqref{problem0:con1}, we omit the subscript of $\mathbf{d}_{IT_2}$ and simply write $\mathbf{d}$ for brevity. 

By exploiting the block-diagonal structure of $\bthe$, problem \eqref{problem0} can be equivalently rewritten as
\begin{subequations}\label{problem:group}
\begin{align}
\max_{\{\bthe_g\}}~&\left|\sum_{g=1}^{G}\h_{RI,g}^H\bthe_g\h_{IT_1,g}\right|^2\\
\text{s.t. }~&\bthe_g\h_{IT_2,g}=\mathbf{d}_g,\quad g=1,2,\dots,G,\label{con1:group}\\
&\bthe_g^H\bthe_g=\mathbf{I}_{G_s},\quad g=1,2,\dots,G.\label{con2:group}
\end{align}
\end{subequations}
Here, $\h_{IT_l}=[\h_{IT_l,1}^T,\h_{IT_l,2}^T,\dots,\h_{IT_l,G}^T]^T$ with $\h_{IT_l,g}\in\C^{G_s\times 1}$,  $l\in\{1,2\}$, $\h_{RI}=[\h_{RI,1}^T, \h_{RI,2}^T,\dots,\h_{RI,G}^T]^T$ with $\h_{RI,g}\in\C^{G_s\times 1}$, and $\mathbf{d}=[\mathbf{d}_1^T,\mathbf{d}_2^T,\dots,\mathbf{d}_G^T]^T$ with $\mathbf{d}_g\in\C^{G_s\times 1}$.
Due to constraints \eqref{con1:group} and \eqref{con2:group}, problem \eqref{problem:group} is feasible only if 
\begin{equation}\label{feasibility}
\|\mathbf{d}_g\|=\|\h_{IT_2,g}\|,\quad g=1,2,\dots,G.
\end{equation}
In the following, we assume that \eqref{feasibility} holds.

Compared with conventional single-operator RIS design \cite{BDRIS,closeform}, the main challenge in \eqref{problem:group} lies in the additional linear constraint \eqref{con1:group}.  A key insight is that, this constraint specifies the behavior of each block $\bthe_g\in\C^{G_s\times G_s}$ only along the direction of $\h_{IT_{2,g}}\in\C^{G_s\times 1}$. When $G_s\geq 2$, the remaining $(G_s-1)$-dimensional subspace orthogonal to $\h_{IT_2,g}$ is still free and can be optimized to improve the desired signal. 
Motivated by this idea, Section \ref{subsec:technique} first develops a novel approach for handling the coupled unitary and linear equality constraints in \eqref{con1:group} and \eqref{con2:group}. Based on this approach, Sections \ref{subsec:fully} and \ref{subsec:group} derive closed-form solutions to problem \eqref{problem:group} for fully-connected RIS and group-connected RIS with  $G_s\geq 2$, respectively. The single-connected RIS case, i.e., $G_s=1$, is discussed in Section \ref{subsec:single}.

\subsection{A Novel Approach For Handling Coupled Unitary and Linear Equality  Constraints}\label{subsec:technique}
Given two vectors $\x\in\C^{n\times 1}$ and $\y\in\C^{n\times 1}$ satisfying $\|\x\|=\|\y\|\neq 0$, this subsection develops a novel approach for reformulating  the following constraints on $\bQ\in\C^{n\times n}$:
\begin{equation}\label{generic_constraint}
\bQ\x=\y,\qquad \bQ^H\bQ=\mathbf{I}_n.
\end{equation}
In problem \eqref{problem:group}, constraints \eqref{con1:group} and \eqref{con2:group} correspond to \eqref{generic_constraint} with
$n=G_s$,  $\x=\h_{IT_2,g}$,  $\y=\mathbf{d}_g,$ and $\bQ=\bthe_g$.
 
We begin by introducing the notion of unitary completion.\begin{definition}[Unitary completion of a vector]\label{definition:UC_vector}
For any nonzero vector $\mathbf{x}\in\mathbb{C}^{n\times 1}$, a unitary matrix $\bU\in\C^{n\times n}$ is called a unitary completion of $\x$ if its first column is $\x/\|\x\|$, i.e.,   $\mathbf{U}\mathbf{e}_1={\mathbf{x}}/{\|\mathbf{x}\|}$, where $\mathbf{e}_1=[1,0,\dots,0]^T$. Such a matrix is denoted by $\bU(\x)$, and further satisfies 
\begin{equation}\label{U_property}
\bU(\x)^H\x=\|\x\|\mathbf{e}_1.
\end{equation}
\end{definition}
\begin{remark}
For a given vector $\x$, its unitary completion can be constructed as follows. The first column is $\x/\|\x\|$, and the remaining columns form an orthonormal basis for the orthogonal complement of $\x$, which can be obtained, for example, by applying the QR decomposition to the orthogonal projection matrix $\bP_{\x}^{\perp}=\mathbf{I}-\frac{\x\x^H}{\|\x\|^2}$.\end{remark}

Using Definition~\ref{definition:UC_vector}, the coupled constraints in \eqref{generic_constraint} can be reformulated in a more convenient form.
\begin{proposition}\label{pro1}
Given $\x\in\C^n$ and  $\y\in\C^n$ with $\|\x\|=\|\y\|\neq 0$. Let $\bU(\x)$ and $\bU(\y)$ be any unitary completion of $\x$ and $\y$, respectively.  The constraint in \eqref{generic_constraint} is satisfied if and only if $\bQ$ can be expressed as 
\begin{equation}\label{decomposition}
\bQ=\bU(\mathbf{y})
\begin{bmatrix}
1 & \mathbf{0}\\
\mathbf{0} & \bar{\bQ}
\end{bmatrix}
\bU(\x)^H,
\end{equation}
where $\bar{\bQ}\in\C^{(n-1)\times (n-1)}$ is any unitary matrix, i.e., $\bar{\bQ}^H\bar{\bQ}=\mathbf{I}_{n-1}$.
\end{proposition}
\begin{proof}
We first prove the necessity. Suppose that $\bQ$ satisfies \eqref{generic_constraint}.
Define
\begin{equation}\label{def:Q_generic}
\widetilde{\bQ}:=\bU(\y)^H\bQ\bU(\x).
\end{equation}
Since $\bU(\y)$, $\bQ$, and $\bU(\x)$ are all unitary, $\widetilde{\bQ}$ is also unitary. Moreover, its first column is given by
\[
\begin{aligned}
\widetilde{\bQ}\mathbf{e}_1=\bU(\y)^H\bQ\bU(\x)\mathbf{e}_1&\overset{(a)}{=}\bU(\y)^H\bQ\frac{\x}{\|\x\|}\\
&\overset{(b)}{=}\frac{1}{\|\x\|}\bU(\y)^H\y\\
&\overset{(c)}{=}\frac{\|\y\|}{\|\x\|}\mathbf{e}_1\\
&\overset{(d)}{=}\mathbf{e}_1,
\end{aligned}
\]
where $(a)$ follows from the definition of unitary completion, $(b)$ uses the constraint $\bQ\x=\y$, $(c)$ follows from \eqref{U_property}, and $(d)$ uses $\|\x\|=\|\y\|$. 
Therefore, $\widetilde{\bQ}$ has the form
\[
\widetilde{\bQ}=
\begin{bmatrix}
1 & \boldsymbol{\beta}^T\\
\mathbf{0} & \bar{\bQ}
\end{bmatrix}
\]
for some vector $\boldsymbol{\beta}\in\C^{n-1}$ and some matrix $\bar{\bQ}\in\C^{(n-1)\times(n-1)}$. Since $\widetilde{\bQ}$ is unitary, we have $\bar{\bQ}\bar{\bQ}^H=\mathbf{I}_{n-1}$ and $1+\|\boldsymbol{\beta}\|^2=1$, i.e. $\boldsymbol{\beta}=\mathbf{0}$. Hence,
\[
\widetilde{\bQ}=
\begin{bmatrix}
1 & \mathbf{0}\\
\mathbf{0} & \bar{\bQ}
\end{bmatrix}.
\]
Substituting $\widetilde{\bQ}$ into \eqref{def:Q_generic} yields \eqref{decomposition},   which proves the necessity.

We next prove the sufficiency. Suppose that $\bQ$ is given by \eqref{decomposition} with $\bar{\bQ}^H\bar{\bQ}=\mathbf{I}_{n-1}$.
Since $\bU(\x)$, $\bU(\y)$, and
$
\left[\begin{smallmatrix}
1 & \mathbf{0}\\
\mathbf{0} & \bar{\bQ}
\end{smallmatrix}\right]$
are all unitary, it follows that $\bQ$ is unitary, i.e.,
$\bQ^H\bQ=\mathbf{I}_n.$
It remains to verify the linear constraint. Using \eqref{decomposition},
\[
\begin{aligned}
\bQ\x
&=\bU(\y)
\begin{bmatrix}
1 & \mathbf{0}\\
\mathbf{0} & \bar{\bQ}
\end{bmatrix}
\bU(\x)^H\x\\
&\overset{(a)}{=}
\bU(\y)
\begin{bmatrix}
1 & \mathbf{0}\\
\mathbf{0} & \bar{\bQ}
\end{bmatrix}
\|\x\|\mathbf{e}_1\\
&=\|\x\|\bU(\y)\mathbf{e}_1\\
&\overset{(b)}{=}\frac{\|\x\|}{\|\y\|}\y\overset{(c)}{=}\y,
\end{aligned}
\]
where $(a)$ follows from \eqref{U_property}, $(b)$ uses the definition of unitary completion, and $(c)$ uses $\|\x\|=\|\y\|$. Therefore, $\bQ\x=\y$ holds. 
The proof is complete.
\end{proof}
Proposition~\ref{pro1} shows that all matrices satisfying \eqref{generic_constraint} can be written in the form of \eqref{decomposition}. In particular, the degrees of freedom in $\bQ$ under constraint \eqref{generic_constraint} are fully captured by the $(n-1)$-dimensional unitary matrix $\bar{\bQ}$. 
From an optimization viewpoint, this result converts the coupled unitary and linear equality constraint on $\bQ$ into a unitary-only constraint on $\bar{\bQ}$, which is much easier to handle. Proposition \ref{pro1} can also be generalized to involve matrix-form linear equality constraints; see Section \ref{subsec:solution_general} for detailed discussions.  In the following, we exploit Proposition \ref{pro1} to solve problem \eqref{problem:group}.
\subsection{Closed-Form Solution for Fully-Connected RIS}\label{subsec:fully}
To illustrate the main idea and ease the notation, we first derive closed-form  globally optimal solutions to \eqref{problem:group} for fully-connected RIS in this subsection,  and then extend the results to general group-connected RIS with $G_s\geq 2$ in Section \ref{subsec:group}.  

For fully-connected RIS,  we have $G=1$ and $G_s=N$, thus problem \eqref{problem:group} reduces to
\begin{subequations}\label{problem:fully}
\begin{align}
\max_{\bthe}~&|\h_{RI}^H\bthe\h_{IT_1}|^2\\
\text{s.t. }~&\bthe\h_{IT_2}=\mathbf{d},\label{problemfully:con1}\\
&\bthe^H\bthe=\mathbf{I}_{N}.\label{problemfully:con2}
\end{align}
\end{subequations}  
By applying Proposition \ref{pro1} to constraints \eqref{problemfully:con1} and \eqref{problemfully:con2},  problem \eqref{problem:fully} can be reformulated as
\begin{subequations}\label{problem:fully2}
\begin{align}
\max_{\{\bthe,\bar{\bthe}\}}&
\left|\h_{RI}^H\bthe\h_{IT_1}\right|^2\label{fully:obj}\\
\text{s.t. }~~&
\bthe=\bU(\mathbf{d})
\begin{bmatrix}
1 & \mathbf{0}\\
\mathbf{0} & \bar{\bthe}
\end{bmatrix}
\bU(\h_{IT_{2}})\hspace{-0.03cm}^H,\label{fully:con1_2}\\
&
\bar{\bthe}^H\bar{\bthe}=\mathbf{I}_{N-1},
\end{align}
\end{subequations}
where $\bU(\mathbf{d})$ and $\bU(\h_{IT_{2}})$ are arbitrary unitary completions of $\mathbf{d}$ and $\h_{IT_{2}}$, respectively. 
Substituting \eqref{fully:con1_2} into the objective function in \eqref{fully:obj}, and denoting $\bw=\bU(\mathbf{d})^H\h_{RI}$ and  $\mathbf{v}=\bU(\h_{IT_{2}})^H\h_{IT_{1}},$
we obtain an equivalent optimization problem over $\bar{\bthe}$ as follows
\begin{subequations}\label{problem:fully3}
\begin{align}
\max_{\bar{\bthe}}~&
\left|\left([\bw]_1^*[\mathbf{v}]_1+[\bw]_{2:N}^H\bar{\bthe}[\mathbf{v}]_{2:N}\right)\right|^2\label{fully:obj2}\\
\text{s.t. }~&
\bar{\bthe}^H\bar{\bthe}=\mathbf{I}_{N-1}.
\end{align}
\end{subequations}
By the Cauchy-Schwarz inequality, the optimal value is
\begin{equation}\label{fully:optobj}
P_R^\star=
\left(
|[\bw]_1^*[\mathbf{v}]_1|
+
\|[\bw]_{2:N}\|\|[\mathbf{v}]_{2:N}\|
\right)^2,
\end{equation}
which is attained when $\bar{\bthe}[\mathbf{v}]_{2:N}$ is aligned with $[\bw]_{2:N}$, and the phase of  $[\bw]_{2:N}^H\bar{\bthe}[\mathbf{v}]_{2:N}$ is the same as that of  $[\bw]_1^*[\mathbf{v}]_1$, i.e., 
\begin{equation}\label{eqn:bartheta}
\bar{\bthe}[\mathbf{v}]_{2:N}
=
e^{\mathrm{i}\angle([\bw]_1^*[\mathbf{v}]_1)}
\frac{\|[\mathbf{v}]_{2:N}\|}{\|[\bw]_{2:N}\|}
[\bw]_{2:N}.
\end{equation}
Applying Proposition \ref{pro1} to \eqref{eqn:bartheta} and noting that $\bU(\alpha\x)=\bU(\x)$ for any $\alpha>0$,   the solution can be expressed as \[
\bar{\bthe}^\star=
e^{\mathrm{i}\angle([\bw]_1^*[\mathbf{v}]_1)}
\bU([\bw]_{2:N})
\begin{bmatrix}
1 & \mathbf{0}\\
\mathbf{0} & \widetilde{\bthe}
\end{bmatrix}
\bU([\mathbf{v}]_{2:N})^H,
\]
where $\widetilde{\bthe}$ is any unitary matrix of dimension $N-2$, i.e., $\widetilde{\bthe}^H\widetilde{\bthe}=\mathbf{I}_{N-2}$.  
The optimal solution to \eqref{problem:fully} is then obtained by substituting $\bar{\bthe}^\star$ into \eqref{fully:con1_2}.

 \begin{remark}
Although there exist infinitely many choices of $\bU(\mathbf{d})$ and $\bU(\h_{IT_{2}})$, the optimization problem in \eqref{problem:fully3} is independent of the specific unitary completions adopted.   To illustrate this, it suffices to show that both $[\bw]_1^*[\mathbf{v}]_1$ and the range of $$f(\bar{\bthe}):=[{\bw}]_{2:N}^H\bar{\bthe}[{\mathbf{v}}]_{2:N},~\bar{\bthe}^H\bar{\bthe}=\mathbf{I}_{N-1}$$ are irrelevant to $\bU(\bd)$ and $\bU(\h_{IT_{2}})$.
 
 Since the first columns of $\bU(\bd)$ and $\bU(\h_{IT_{2}})$ are fixed as $\bd/\|\bd\|$ and $\h_{IT_{2}}/\|\h_{IT_{2}}\|$, respectively, we have 
\begin{equation}\label{firstterm}
\begin{aligned}
[\bw]_1=\frac{\bd^H\h_{RI}}{\|\bd\|}\text{ and }
[\mathbf{v}]_1=\frac{\h_{IT_{2}}^H\h_{IT_{1}}}{\|\h_{IT_{2}}\|},
\end{aligned}
\end{equation}
which are independent of $\bU(\bd)$ and $\bU(\h_{IT_{2}})$, and so is the product $[\bw]_1^*[\mathbf{v}]_1$.
In addition, the range of $f(\bar{\bthe})$ can be expressed as   
 $$
 \begin{aligned}
 \operatorname{Range}(f)&=\{[{\bw}]_{2:N}^H\bar{\bthe}[{\mathbf{v}]_{2:N}}\mid \bar{\bthe}^H\bar{\bthe}=\mathbf{I}_{N-1}\}\\
 &=\{x\mid |x|\leq\|[{\bw}]_{2:N}\|\|[\mathbf{v}]_{2:N}\| \},
 \end{aligned}$$
where  
$$
\begin{aligned}
\|[{\bw}]_{2:N}\|^2&=\|\bU(\mathbf{d})^H\h_{RI}\|^2-\left|[\bU(\mathbf{d})^H\h_{RI}]_{1}\right|^2\\
&\overset{(a)}=\|\h_{RI}\|^2-\frac{|\bd^H\h_{RI}|^2}{\|\bd\|^2}
\end{aligned}$$ and
 $$
 \begin{aligned}
 \|[{\mathbf{v}}]_{2:N}\|^2
 &=\|\bU(\mathbf{h}_{IT_{2}})^H\h_{IT_{1}}\|^2-\left|[\bU(\mathbf{h}_{IT_{2}})^H\h_{IT_{1}}]_{1}\right|^2\\
 &\overset{(b)}=\|\h_{IT_{1}}\|^2-\frac{|\h_{IT_{2}}^H\h_{IT_{1}}|^2}{\|\h_{IT_{2}}\|^2}
 \end{aligned}$$
 are irrelevant to the choices of $\bU(\bd)$ and $\bU(\h_{IT_{2}})$; {(a)} and (b) use the fact that   $\bU(\bd)$, $\bU(\h_{IT_{2}})$ are unitary and \eqref{firstterm}.
 
   In particular, the optimal value  in \eqref{fully:optobj}  is irrelevant to the specific choices of $\bU(\bd)$ and $\bU(\h_{IT_{2}})$, which can be expressed as
 \begin{equation*}\label{fstar:fully}
\begin{aligned}
&P_R^\star=\Bigg(\frac{\left|\bd^H\h_{RI}\right|}{\|\bd\|}\frac{\left|\h_{IT_2}^H\h_{IT_1}\right|}{\|\h_{IT_2}\|}\\
&+\hspace{-0.05cm}\bigg(\|\h_{RI}\|^2\hspace{-0.05cm}-\hspace{-0.05cm}\frac{|\bd^H\h_{RI}|^2}{\|\bd\|^2}\bigg)^{\frac{1}{2}}\bigg(\|\h_{IT_1}\|^2\hspace{-0.05cm}-\hspace{-0.05cm}\frac{|\h_{IT_2}^H\h_{IT_1}|^2}{\|\h_{IT_2}\|^2}\bigg)^{\frac{1}{2}}\Bigg)^2\hspace{-0.05cm}.
\end{aligned}
\end{equation*}
\end{remark}

\subsection{Closed-Form Solution for Group-Connected RIS with $G_s\geq 2$}\label{subsec:group}
In this subsection, we extend the previous approach  to general group-connected RIS with  $G_s\geq 2$. 

According to Proposition \ref{pro1},  constraints \eqref{con1:group} and \eqref{con2:group} are equivalent to
$$\bthe_g=\bU(\mathbf{d}_g)\left[\begin{matrix}1&\mathbf{0}\\\mathbf{0}&\bar{\bthe}_g\end{matrix}\right]\bU(\h_{IT_2,g})^H,~\bar{\bthe}_g^H\bar{\bthe}_g=\mathbf{I}_{G_s-1},$$for all $g=1,2,\dots, G$.
 Utilizing this and denoting  ${\bw}_g=\bU(\mathbf{d}_g)^H\h_{RI,g}$ and ${\mathbf{v}}_g=\bU(\h_{IT_2,g})^H\h_{IT_1,g}$,  problem \eqref{problem:group} transforms to 
 a problem over $\{\bar{\bthe}_g\}_{1\leq g\leq G}$ as follows
\begin{subequations}\label{problem5}
\begin{align}
\max_{\{\bar{\bthe}_g\}}~&
\left|\sum_{g=1}^G\left([\bw_g]_1^*[\mathbf{v}_g]_1+[\bw_g]_{2:G_s}^H\bar{\bthe}_g[\mathbf{v}_g]_{2:G_s}\right)\right|^2\label{obj2}\\
\text{s.t. }~&
\bar{\bthe}_g^H\bar{\bthe}_g=\mathbf{I}_{G_s-1},\quad g=1,2,\dots,G.
\end{align}
\end{subequations}
Let $\gamma=\sum_{g=1}^G[\bw_g]_1^*[\mathbf{v}_g]_1.$  Similar to the discussions in Section \ref{subsec:fully},  
 the optimal solution is
\[
\bar{\bthe}_g^\star=
e^{\mathrm{i}\angle\gamma}
\bU([\bw_g]_{2:G_s})
\begin{bmatrix}
1 & \mathbf{0}\\
\mathbf{0} & \widetilde{\bthe}_g
\end{bmatrix}
\bU([\mathbf{v}_g]_{2:G_s})^H,
\]
where $\widetilde{\bthe}_g^H\widetilde{\bthe}_g=\mathbf{I}_{G_s-2}$ if $G_s>2$, while the above expression reduces to a scalar phase factor when $G_s=2$. The optimal value is 
 \begin{equation}\label{group:optvalue}
\begin{aligned} 
P_R^\star
=&\Bigg(\left|\sum_{g=1}^G\frac{(\bd_g^H\h_{RI,g})^{*}}{\|\bd_g\|}\frac{\h_{IT_2,g}^H\h_{IT_1,g}}{\|\h_{IT_2,g}\|}\right|\\
&~~+\sum_{g=1}^G\bigg(\|\h_{RI,g}\|^2-\frac{|\bd_g^H\h_{RI,g}|^2}{\|\bd_g\|^2}\bigg)^{\frac{1}{2}}\\
&~~~~~~\times\bigg(\|\h_{IT_1,g}\|^2-\frac{|\h_{IT_2,g}^H\h_{IT_1,g}|^2}{\|\h_{IT_2,g}\|^2}\bigg)^{\frac{1}{2}}\Bigg)^2.
\end{aligned}
\end{equation}
The above framework can also be applied to the case where the direct path between BS$_1$ and the user is present. The only minor difference is that the definition of $\gamma$ should be modified to include the direct path $h_{RT}$, i.e.,
$
\gamma = h_{RT}+\sum_{g=1}^G[\bw_g]_1^*[\mathbf{v}_g]_1.
$

\subsection{Closed-Form Solution for Single-Connected RIS}\label{subsec:single}

For a single-connected RIS, we have $G=N$ and $G_s=1$. In this case, the scattering matrix is uniquely determined by the constraint $\bthe\h_{IT_2}=\mathbf{d}$, and its diagonal entries are given by
\begin{equation}\label{theta:single}
[\bthe^\star]_{n,n}
=
e^{\mathrm{i}\angle([\mathbf{d}]_n)-\mathrm{i}\angle([\h_{IT_2}]_n)},
\qquad n=1,2,\dots,N.
\end{equation}
Accordingly, the resulting objective value is
$$P_R^\star=\left|\h_{RI}^H\text{diag}(e^{\mathrm{i}\angle(\bd-\h_{IT_2})})\h_{IT_1}\right|^2.$$
Hence, all RIS degrees of freedom are exhausted by the constraint, leaving no extra flexibility to improve the desired signal and leading to poor  performance; see also discussions in Section \ref{sec:scalinglaw}.

\section{Scaling Law Analysis of the Two-Operator System}\label{sec:scalinglaw}
In this section, we characterize the fundamental limits of different RIS architectures in the two-operator system considered in Section \ref{sec:solution}. Specifically, we derive the scaling laws that quantify how the received signal power grows with the number of RIS elements for fully-, group-, and single-connected RIS,  based on the optimal values obtained in Section \ref{sec:solution}. 

Throughout this section, we assume that the RIS-user channel is i.i.d. Rayleigh fading, i.e., $\h_{RI}\sim\mathcal{CN}(\mathbf{0},\rho_{RI}\mathbf{I}_{N})$, where $\rho_{RI}$ is the channel gain. The Rayleigh fading model accounts for the rich scattering environment and the mobility of the user, which is commonly adopted in the existing literature \cite{RIS_tutorial}.

The BS-RIS channels are typically modeled as Rician fading, capturing the presence of a dominant LoS component due to their quasi-static nature. For analytical tractability, we characterize in the following the received signal power under both Rayleigh fading and LoS BS-RIS channels, which correspond to two extreme cases of Rician fading. The  Rician BS-RIS channels are evaluated in the simulations, and the results exhibit the same scaling behavior as these two extreme cases; see Section \ref{sec:simulation}.

The following theorem characterizes the scaling law under Rayleigh fading BS-RIS channels, modeled as $\h_{IT_l}\sim\mathcal{CN}(\mathbf{0},\rho_{IT_l}\mathbf{I}_N)$, where $\rho_{IT_l}$  denotes the channel gain,~$l\in\{1,2\}$.
 \begin{theorem}[Expected Received Signal Power Under Rayleigh Fading BS-RIS Channels]\label{scalinglaw}
Assume that $\h_{RI}\sim\mathcal{CN}(\mathbf{0},\rho_{RI}\mathbf{I}_{N})$ and $\h_{IT_l}\sim\mathcal{CN}(\mathbf{0},\rho_{IT_l}\mathbf{I}_{N}),~l\in\{1,2\}$, all of which are independent.  The following results hold for the optimal received signal power of fully-, group-, and single-connected RIS-aided system.
\begin{itemize}
\item[(i)] For fully-connected RIS, 
$$
\mathbb{E}[P_R^\star]
\hspace{-0.05cm}=\hspace{-0.1cm}\left((N-1)^2\hspace{-0.05cm}+\hspace{-0.05cm}\frac{\pi}{2}\left(\frac{\Gamma(N-\frac{1}{2})}{\Gamma(N-1)}\right)^2+1\right)\rho_{RI}\rho_{IT_1}.
$$
\item[(ii)] For group-connected RIS with group size $G_s\geq 2$,
$$
\begin{aligned}
\mathbb{E}[P_R^\star]
=\bigg(&G(G-1)\left(\frac{\Gamma(G_s-\frac{1}{2})}{\Gamma(G_s-1)}\right)^4\\
&+\sqrt{\pi}G\frac{\Gamma(G+\frac{1}{2})}{\Gamma(G)}\left(\frac{\Gamma(G_s-\frac{1}{2})}{\Gamma(G_s-1)}\right)^2\\
&+G(G_s-1)^2+G\bigg)\rho_{RI}\rho_{IT_1}.
\end{aligned}$$
\item[(iii)] For single-connected RIS,
$$\mathbb{E}[P_R^\star]=N{\rho_{RI}\rho_{IT_1}}.$$

\end{itemize}
\end{theorem}
\begin{proof}
See Appendix \ref{proofscalinglaw}.
\end{proof}

We further characterize the scaling law under LoS BS–RIS channels. For simplicity, we assume that the RIS is deployed as a uniform linear array (ULA) with half-wavelength antenna spacing. The BS-RIS channels are then modeled as $$\h_{IT_l}=\sqrt{\rho_{IT_l}}\mathbf{a}(\theta_{IT_l}),$$ where $\mathbf{a}(\theta)=[1,e^{-\mathrm{i}\pi\sin\theta},\dots,e^{-\mathrm{i}(N-1)\pi\sin\theta}]^T$ denotes the steering vector corresponding to angle $\theta$. The angle $\theta_{IT_l}\in(-\frac{\pi}{2},\frac{\pi}{2})$ represents the angle of arrival (AOA) from BS$_l$ to RIS, and $\rho_{IT_l}$ denotes the channel gain, $l\in\{1,2\}$.
\begin{theorem}[Expected Received Signal Power Under LoS BS-RIS Channels]\label{scalinglaw2}
Assume that $\h_{RI}\sim\mathcal{CN}(\mathbf{0},\rho_{RI}\mathbf{I}_{N})$ and $\h_{IT_l}=\mathbf{a}(\theta_{IT_l})$, where $\theta_{IT_l}\in(-\frac{\pi}{2},\frac{\pi}{2}),~l\in\{1,2\},$ with $\theta_{IT_1}\neq \theta_{IT_2}$. Let $\mu_l=\pi\sin\theta_{IT_l},~l\in\{1,2\},$ and $\Delta\mu=\mu_2-\mu_1$. In addition, define the following quantity that is related to the group size $G_s$ and $\Delta\mu$:
\begin{equation}\label{def:L_G}
L_{G_s,\Delta\mu}:=\frac{1}{G_s}\left(\frac{\sin\frac{G_s\Delta\mu}{2}}{\sin\frac{\Delta\mu}{2}}\right)^2.
\end{equation} Then, the following results hold for the optimal received signal power of fully-, group-, and single-connected RIS-aided system.
\begin{itemize}
\item[(i)] For fully-connected RIS, 
$$
\begin{aligned}
\mathbb{E}[P_R^\star]
=&\rho_{RI}\rho_{IT_1}\bigg((N-1)(N-L_{N,\Delta\mu})\\
+&\sqrt{\pi(N-L_{N,\Delta\mu})L_{N,\Delta\mu}}\frac{\Gamma(N-\frac{1}{2})}{\Gamma(N-1)}+L_{N,\Delta\mu}\bigg).
\end{aligned}$$
\item[(ii)] For group-connected RIS with group size $G_s\geq 2$,
$$
\begin{aligned}
\hspace{-0.7cm}\mathbb{E}[P_R^\star]=&\bigg(G(G-1)\left(\frac{\Gamma(G_s-\frac{1}{2})}{\Gamma(G_s-1)}\right)^2(G_s-L_{G_s,\Delta\mu})\\
&+G\sqrt{\pi G(G_s-L_{G_s,\Delta\mu})L_{G_s,\Delta\mu}}\frac{\Gamma(G_s-\frac{1}{2})}{\Gamma(G_s-1)}\\
&+G(G_s-1)(G_s-L_{G_s,\Delta\mu})+GL_{G_s,\Delta\mu}\bigg){\rho_{RI}\rho_{IT_1}}.
\end{aligned}$$
\item[(iii)] For single-connected RIS,
$$\mathbb{E}[P_R^\star]=N{\rho_{RI}\rho_{IT_1}}.$$

\end{itemize}
\end{theorem}
\begin{proof}
{\color{black}The proof follows the same procedure as that of Theorem \ref{scalinglaw}; see Appendix C in the supplemental material. }
\end{proof}
\begin{remark}[Large-$N$ Scaling Law]\label{remark:scaling}
From Theorems \ref{scalinglaw} and \ref{scalinglaw2}, the expected received signal power exhibits distinct asymptotic scaling behaviors for different RIS architectures in the large-$N$ regime.
\begin{itemize}
\item[(i)] Fully-connected RIS:
$$\mathbb{E}[P_R^\star]\overset{N\to\infty}\sim N^2{\rho_{RI}\rho_{IT_1}}.$$
\item[(ii)] Group-connected RIS with group size $G_s\geq 2$: 
$$
\mathbb{E}[P_R^\star]\overset{N\to\infty}{\sim}\kappa N^2\rho_{RI}\rho_{IT_1},
$$
where the scaling coefficient $\kappa>0$ depends on the group size $G_s$. Specifically,  $\kappa=\left(\frac{1}{\sqrt{G_s}}\frac{\Gamma(G_s-\frac{1}{2})}{\Gamma(G_s-1)}\right)^4$ for Rayleigh fading BS-RIS channels, and $\kappa=\left(\frac{1}{\sqrt{G_s}}\frac{\Gamma(G_s-\frac{1}{2})}{\Gamma(G_s-1)}\right)^2$ for LoS BS-RIS channels. 
\item[(iii)] Single-connected RIS: $$\mathbb{E}[P_R^\star]=N{\rho_{RI}\rho_{IT_1}}.$$
\end{itemize}
\end{remark}
We can make the following conclusions regarding the scaling behaviors of different RIS architectures. 

For single-connected RIS, the received signal power increases \emph{linearly} with the number of RIS elements $N$ in the two-operator system.  This sharply contrasts with the single-operator scenario where the received signal power scales quadratically with $N$ \cite{RIS1}. The reason is that a single-connected RIS lacks sufficient degrees of freedom to simultaneously enhance the desired signal and maintain a fixed RIS-reflected channel for the non-serving operator, as discussed at the end of Section~\ref{sec:solution}. Interestingly, a similar linear behavior has also been reported in  \cite{Miridakis_2024_multi} and \cite{Miridakis_2026_multi_2}, though derived from different perspectives and based on different assumptions.  

For group-connected RIS with $G_s\geq 2$ (including fully-connected RIS as a special case),  the received signal power scales \emph{quadratically} with $N$ under both Rayleigh fading and LoS BS-RIS channels. The group size mainly affects the scaling coefficient, which is characterized by  $\kappa$ in Remark \ref{remark:scaling}. As shown in \cite[Eq. (56)]{BDRIS},  under Rayleigh fading, the received signal power scales quadratically in a single-operator system, with scaling coefficient $\kappa_o= \left(\frac{1}{\sqrt{G_s}}\frac{\Gamma(G_s+\frac{1}{2})}{\Gamma(G_s)}\right)^4$. It then follows that $$\mathbb{E}[P_R^\star]\overset{N\to\infty}\sim \left(1-\frac{1}{2G_s-1}\right)^4\mathbb{E}[P_R^{{o}}]$$ under Rayleigh fading, 
where $\mathbb{E}[P_R^{{o}}]$ denotes the expected received signal power in the single-operator system, and we have used the identity $\Gamma(x+1)=x\Gamma(x)$.  This implies that the same asymptotic scaling law as the single-operator case is preserved in the two-operator system, with only a multiplicative loss factor of $(1-\frac{1}{2G_s-1})^4$. Since this factor approaches one as $G_s$ increases, the asymptotic performance approaches that of the single-operator case for large $G_s$. 


To conclude, in the two-operator systems, the performance of conventional single-connected RIS degrades dramatically from $\mathcal{O}(N^2)$ to $\mathcal{O}(N)$ due to its limited flexibility. In contrast, a group-connected RIS with group size two, which is among the simplest BD-RIS architectures, already provides sufficient flexibility to preserve the desirable $\mathcal{O}(N^2)$ gain as in the single-operator scenario.

{\section{Optimization and Scaling Laws for General Multi-Operator Systems}\label{sec:multioperator}}
In the previous sections, we considered a two-operator system with only one non-serving operator. In this section, we extend the proposed framework to the general multi-operator setting in \eqref{problem00}. 

We decompose $\h_{IT_l}$ and $\d_{IT_l}$ as $
\h_{IT_l} = [\h_{IT_l,1}^T,\h_{IT_l,2}^T,\dots,\h_{IT_l,G}^T]^T$ with  $\h_{IT_l,g}\in\C^{G_s\times 1}$,
and
$\d_{IT_l} = [\d_{IT_l,1}^T,\d_{IT_l,2}^T,\dots,\d_{IT_l,G}^T]^T$ with $\d_{IT_l,g}\in\C^{G_s\times 1}$.
The constraints \eqref{con01} and \eqref{con02} can then be rewritten as 
\begin{equation}\label{constraint:linear_multioperator}
\bthe_g\mathbf{H}_g=\mathbf{D}_g,~g=1,2,\dots, G,
\end{equation}
where $\bH_g=[\h_{IT_2,g},\h_{IT_3,g},\dots,\h_{IT_L,g}]\in\C^{G_s\times (L-1)}$ and $\mathbf{D}_g=[\d_{IT_2,g},\d_{IT_3,g},\dots,\d_{IT_L,g}]\in\C^{G_s\times (L-1)}$. Due to the unitary constraint on $\bthe_g$, the problem is feasible if and only if  
\begin{equation}\label{feasible_condition}
\bH_g^H\bH_g=\bD_g^H\bD_g,~g=1,2,\dots, G.
\end{equation}
We next focus on the nontrivial case that \eqref{feasible_condition}  is satisfied. 
\vspace{0.2cm}
\subsection{Closed-Form Solutions to Problem \eqref{problem00}}\label{subsec:solution_general}
In this subsection, we give the closed-form solutions to problem \eqref{problem00}. We consider two cases, $ G_s\geq L$ (which includes fully-connected RIS as a special case) and $G_s< L$ (which includes single-connected RIS as a special case),  separately.
\subsubsection{$G_s\geq L$} The technique developed in Section \ref{subsec:technique} can be extended to solve \eqref{problem00} for $G_s\geq L$. First, we extend Definition \ref{definition:UC_vector} to the matrix case. 
\begin{definition}[Unitary completion of a matrix]\label{def:UC_matrix}
Given a full-rank matrix $\mathbf{X}\in\mathbb{C}^{n\times m}$ with $m\leq n $, a unitary matrix $\bU\in\C^{n\times n}$ is called a unitary completion of $\bX$ if it satisfies  $$\mathbf{U}\mathbf{E}_m=\bX(\bX^H\bX)^{-\frac{1}{2}},$$ where $\mathbf{E}_m=[\mathbf{e}_1,\mathbf{e}_2,\dots,\mathbf{e}_m]\in\R^{n\times m}$. Such a matrix is denoted by $\bU(\bX)$.
\end{definition}
With the above definition, we can rewrite the constraints in problem \eqref{problem00} into more tractable forms. 
\begin{proposition}\label{pro4}
Assume that $\bH_g^H\bH_g=\bD_g^H\bD_g$ for $g=1,2,\dots, G_s$, and are of full rank. Let $\bU(\bH_g)$ and $\bU(\bD_g)$ be any unitary completions of $\bH_g$ and $\bD_g$, respectively.   Constraints \eqref{con03} and \eqref{constraint:linear_multioperator} are satisfied if and only if $\bthe_g$ can be expressed as 
\begin{equation*}
\bthe_g=\bU(\mathbf{D}_g)\left[\begin{matrix}\mathbf{I}_{L-1}&\mathbf{0}\\\mathbf{0}&\bar{\bthe}_g\end{matrix}\right]\bU(\bH_g)^H,~g=1,2,\dots, G,
\end{equation*}
where $\bar{\bthe}_g^H\bar{\bthe}_g= \mathbf{I}_{G_s-L+1}$. 
\end{proposition}
\hspace{-0.35cm}The proof of Proposition \ref{pro4} follows the same steps as that of Proposition \ref{pro1}, which relies on the definition and properties of the unitary completion given in Definition \ref{def:UC_matrix}. For completeness, we provide a proof in Appendix D of the supplemental material. 

Applying Proposition \ref{pro4} and rewriting the objective function of \eqref{problem00} as in \eqref{problem:group}, problem  \eqref{problem00} transforms to    \begin{equation}\label{problem6}
\begin{aligned}
\max_{\{\bar{\bthe}_g\}}~&\left|\sum_{g=1}^G\left([\bw_g]_{1:L-1}^H[{\mathbf{v}}_g]_{1:L-1}+{[{\bw}_g]_{L:G_s}^H\bar{\bthe}_g[{\mathbf{v}}_g]_{L:G_s}}\right)\right|^2\\
\text{s.t. }~&\bar{\bthe}_g^H\bar{\bthe}_g=\mathbf{I}_{G_s-L+1},~g=1,2,\dots,G,
\end{aligned}
\end{equation}
where $\bw_g=\bU(\bD_g)^H\h_{RI,g}$ and $\v_g=\bU(\bH_g)^H\h_{IT_1,g}$ for all $g=1,2,\dots,G$. 
Similar to the procedures in Section \ref{subsec:group}, we get the following optimal solutions of problem \eqref{problem6}:
  $$\bar{\bthe}_g^\star=e^{\mathrm{i}\angle\gamma}\bU([\bw_g]_{L:G_s})\left[\begin{matrix}1&\mathbf{0}\\\mathbf{0}&\widetilde{\bthe}_g\end{matrix}\right]\bU_g([\mathbf{v}_g]_{L:G_s})^H,$$
where $\gamma=\sum_{g=1}^G[\bw_g]_{1:L-1}^H[{\mathbf{v}}_g]_{1:L-1}$ and $\widetilde{\bthe}_g$ is any $G_s-L$ dimensional unitary matrix, i.e., $\widetilde{\bthe}_g^H\widetilde{\bthe}_g=\mathbf{I}_{G_s-L}$,~$g=1,2,\dots,G$. The optimal value is given by
\begin{equation}\label{group:multioperator}
\begin{aligned} 
P_R^\star=&\Bigg(\left|\sum_{g=1}^G\h_{RI,g}^H\bD_g(\bD_g^H\bD_g)^{-\frac{1}{2}}(\bH_g^H\bH_g)^{-\frac{1}{2}}\bH_g^H\h_{IT_1,g}\right|\\
&~~+\sum_{g=1}^G\big(\|\h_{RI,g}\|^2-\|(\bD_g^H\bD_g)^{-\frac{1}{2}}\bD_g^H\h_{RI,g}\|^2\big)^{\frac{1}{2}}\\
&~~~~~~\times\big(\|\h_{IT_1,g}\|^2-\|(\bH_g^H\bH_g)^{-\frac{1}{2}}\bH_g^H\h_{IT_1,g}\|^2\big)^{\frac{1}{2}}\Bigg)^2,
\end{aligned}
\end{equation}
where we have used $[\bw_{g}]_{1:L-1}=(\bD_g^H\bD_g)^{-\frac{1}{2}}\bD_g^H\h_{RI,g}$ and $[\v_g]_{1:L-1}=(\bH_g^H\bH_g)^{-\frac{1}{2}}\bH_g^H\h_{IT_1,g}$, which follow from Definition \ref{def:UC_matrix} and the definitions of $\bw_g$ and $\v_g$.
 \subsubsection{$G_s<L$}
If $G_s<L$, the problem in \eqref{problem00} has a unique feasible solution under \eqref{feasible_condition},  given by
\begin{equation}\label{solution:gs<l}
\bthe_g^\star=\bD_g\bH_g^H(\bH_g\bH_g^H)^{-1},~g=1,2,\dots,G,
\end{equation}
which solves \eqref{constraint:linear_multioperator} and is always unitary because of the feasibility condition \eqref{feasible_condition}.
The corresponding objective value is 
\begin{equation}\label{multi:Gs<L}
P_R^\star=\big|\sum_{g=1}^G\h_{RI,g}^H\bD_g\bH_g^H(\bH_g\bH_g^H)^{-1}\h_{IT_1,g}\big|^2.
\end{equation}

\vspace{-0.5cm}
\subsection{Scaling Law Analysis}\label{subsec:scalinglaw}
We now provide scaling law analysis of different RIS architectures in the general multi-operator system, in order to characterize how the number of operators affects the scaling behavior.

For analytical tractability, we assume that both the BS-RIS channels and the RIS-user channel follow i.i.d. Rayleigh fading, as in Theorem \ref{scalinglaw}. The LoS BS-RIS case exhibits the same scaling behavior, but does not admit a clean formula as in Theorem \ref{scalinglaw2}. We will numerically investigate the LoS and Rician BS-RIS channels in the simulations. The following theorem presents the scaling law of the received signal power for the considered multi-operator systems. 
 \begin{theorem}[Expected Received Signal Power in Multi-Operator Systems]\label{scalinglaw3}
Assume that $\h_{RI}\sim\mathcal{CN}(\mathbf{0},\rho_{RI}\mathbf{I}_{N})$ and $\h_{IT_l}\sim\mathcal{CN}(\mathbf{0},\rho_{IT_l}\mathbf{I}_{N}),~l\in\{1,2,\dots,L\}$, all of which are independent. The following results hold for the optimal received signal power of fully- and group-connected RIS-aided system.
\begin{itemize}
\item[(i)] For fully-connected RIS, 
$$
\begin{aligned}
\hspace{-0.15cm}\mathbb{E}[P_R^\star]\hspace{-0.05cm}
=\hspace{-0.1cm}\bigg((N\hspace{-0.05cm}-\hspace{-0.05cm}L\hspace{-0.05cm}+\hspace{-0.05cm}1)^2+&\sqrt{\pi}\frac{\Gamma(L\hspace{-0.05cm}-\hspace{-0.05cm}\frac{1}{2})}{\Gamma(L\hspace{-0.05cm}-\hspace{-0.05cm}1)}\left(\frac{\Gamma(N\hspace{-0.05cm}-\hspace{-0.05cm}L+\frac{3}{2})}{\Gamma(N\hspace{-0.05cm}-\hspace{-0.05cm}L+1)}\right)^{\hspace{-0.05cm}2}\\
&+L-1\bigg)\rho_{RI}\rho_{IT_1}.
\end{aligned}$$
\item[(ii)] For group-connected RIS with group size $G_s\geq L$,
$$
\begin{aligned}
\hspace{-0.15cm}\mathbb{E}[P_R^\star]
=\hspace{-0.1cm}\bigg(&G(G-1)\left(\frac{\Gamma(G_s-L+\frac{3}{2})}{\Gamma(G_s-L+1)}\right)^4\\
&\hspace{-0.05cm}+\sqrt{\pi}G\frac{\Gamma(G(L-1)+\frac{1}{2})}{\Gamma(G(L-1))}\left(\frac{\Gamma(G_s-L+\frac{3}{2})}{\Gamma(G_s-L+1)}\right)^{\hspace{-0.05cm}2}\\
&+G(G_s-L+1)^2+G(L-1)\bigg)\rho_{RI}\rho_{IT_1}.
\end{aligned}$$
\item[(iii)] For group-connected RIS with group size $G_s<L$,
$$\mathbb{E}[P_R^\star]=N{\rho_{RI}\rho_{IT_1}}.$$
\end{itemize} 
\end{theorem}
\begin{proof}
{\color{black}See Appendix E in the supplemental material.}
\end{proof}
\begin{remark}[Large-$N$ Scaling Law in General Multi-Operator Systems]\label{remark:scaling3}
From Theorem \ref{scalinglaw3}, the expected received signal power exhibits the following asymptotic scaling law:
\begin{itemize}
\item[(i)] Fully-connected RIS:
$$\mathbb{E}[P_R^\star]\overset{N\to\infty}\sim N^2{\rho_{RI}\rho_{IT_1}}.$$
\item[(ii)] Group-connected RIS with group size $G_s\geq L$: 
$$
\mathbb{E}[P_R^\star]\overset{N\to\infty}{\sim}\kappa N^2\rho_{RI}\rho_{IT_1},
$$
where  $\kappa=\left(\frac{1}{\sqrt{G_s}}\frac{\Gamma(G_s-L+\frac{3}{2})}{\Gamma(G_s-L+1)}\right)^4$.

\item[(iii)] Group-connected RIS with group size $G_s<L$: $$\mathbb{E}[P_R^\star]=N{\rho_{RI}\rho_{IT_1}}.$$
\end{itemize}
\end{remark}

Theorem \ref{scalinglaw3} and Remark \ref{remark:scaling3} reveal a scaling-law transition at $G_s=L$: when $G_s\geq L$, the scaling law is $\mathcal{O}(N^2)$, whereas when $G_s<L$, it reduces to $\mathcal{O}(N)$. This  generalizes the conclusion in Section \ref{sec:scalinglaw}. 
When $L=2$, Theorem \ref{scalinglaw3} reduces to Theorem \ref{scalinglaw}. 

The result demonstrates that, as the number of non-serving operators increases, a larger group size is required to achieve quadratic scaling. This is because more degrees of freedom are needed to guarantee fixed RIS-reflected channels for the non-serving operators, i.e., to satisfy  \eqref{con01}. Essentially, quadratic scaling is achieved only when the RIS retains additional flexibility for optimizing the desired signal after satisfying constraint \eqref{con01}, which occurs when $G_s \geq L$. Moreover, by comparing with the scaling law in the single-operator system~\cite[Eq. (56)]{BDRIS}, we obtain, for group-connected RIS with $G_s \geq L$,
\begin{equation}\label{wrtsingle}
\mathbb{E}[P_R^\star]\overset{N\to\infty}\sim \prod_{j=0}^{L-2}\left(1-\frac{1}{2(G_s-L+j)+3}\right)^4\mathbb{E}[P_R^o],\vspace{-0.1cm}
\end{equation}
where $\mathbb{E}[P_R^o]$ is the expected received signal power in the single-operator system. This implies that, although quadratic scaling is preserved for $G_s\geq L$, the performance degradation compared with the single-operator system becomes more severe as $L$ increases. On the other hand, for fixed $L$, the gap narrows as the group size increases.

\vspace{-0.25cm}
\section{Extensions to Multi-Antenna Multi-User Systems}\label{subsec:multiantenna}
In this section, we discuss the extensions of the proposed technique  and results to more general scenarios involving multi-antenna BSs and multiple users.

Assume that the BSs associated with all operators are equipped with $N_T$ transmit antennas. The serving operator, i.e., operator 1, serves a general multi-user system. We only assume that the total number of receive antennas on the user side is $N_R$, while both the number of users and the number of antennas at each user are left unspecified. Let $\bH_{IT_l}\in\C^{N\times N_T}$ denote the channel matrix from BS$_l$ to the RIS for $l=1,2,\dots,L$, and let $\bH_{RI}\in\C^{N_R\times N}$ denote the channel matrix from the RIS to the user side served by operator 1. The effective channel at the user(s) is then given by
$
\bH(\bthe)=\bH_{RI}\bthe\bH_{IT_1},
$
where the notation $\bH(\bthe)$ is used to indicate the dependence of the effective channel on $\bthe$. To eliminate the effect of the RIS on the non-serving operators, we still impose fixed RIS-reflected channel constraints:
$$
\bthe\bH_{IT_l}=\mathbf{D}_{IT_l}\in\C^{N\times N_T},\quad l=2,3,\dots,L. 
$$
Compared to \eqref{con01},  $\mathbf{D}_{IT_l}$ is now a matrix rather than a vector, as the BSs associated with the non-serving operators are equipped with multiple antennas.

Consider the following general utility-maximization problem:
\begin{subequations}\label{problem:general}
\begin{align}
\max_{\bthe, \mathcal V\in\mathcal X}\quad & u\left(\bH(\bthe),\mathcal V\right) \\
\text{s.t.}\quad~~
& \bthe\bH_{IT_l}=\bD_{IT_l}, ~ l=2,3,\dots, L, \label{con1:general}\\
& \bthe=\operatorname{diag}(\bthe_1,\bthe_2,\dots,\bthe_G),  \label{con2:general}\\
& \bthe_g^H\bthe_g=\bI_{G_s}, ~ g=1,2,\dots,G,
\end{align}
\end{subequations}
where $\mathcal V$ collects the transceiver variables associated with the serving operator (e.g., the precoding and combining matrices), $\mathcal X$ denotes the corresponding feasible set, and $u(\cdot,\cdot)$ is a general utility function. This formulation encompasses several important scenarios as special cases. Two representative examples are given below.
\begin{example}[Received signal power maximization for single-stream point-to-point MIMO]\label{example1}
In this case, $\mathcal V=(\mathbf{f}_t,\mathbf{f}_r)$, where $\mathbf{f}_t\in\C^{N_T\times 1}$ denotes the transmit beamformer and $\mathbf{f}_r\in\C^{N_R\times 1}$ is the  receive combiner, which satisfy the transmit power constraint $\|\mathbf{f}_t\|^2\leq P_T$ and combiner normalization constraint $\|\mathbf{f}_r\|=1$, respectively. 
The utility function is 
$
u\!\left(\bH(\bthe),\mathcal V\right)
= \left|\mathbf{f}_r^H \bH(\bthe)\mathbf{f}_t\right|^2.
$
\end{example}
\begin{example}[Sum-rate maximization for multi-user MISO]\label{example2}  In this case, each user is equipped with a single antenna, so that $N_R$ equals the number of users. 
Let $\mathcal{V}=\mathbf{F}=[\mathbf{f}_1,\mathbf{f}_2,\dots,\mathbf{f}_{N_R}] \in \mathbb{C}^{N_T\times N_R}$ be the transmit beamforming matrix, which satisfies the power constraint $\|\mathbf{F}\|_F^2 \le P_T$. 
Let $\h_k^H(\bthe)$ be the $k$-th row of $\bH(\bthe)$, the sum-rate is \[
u\!\left(\bH(\bthe),\mathcal V\right)
=
\sum_{k=1}^{N_R}
\log_2\!\left(
1+\frac{|\bh_k^H(\bthe)\mathbf{f}_k|^2}
{\sum_{j\neq k} |\bh_k^H(\bthe)\mathbf{f}_j|^2+\sigma^2}
\right),
\]
where $\sigma^2$ is the noise power.\end{example}


To solve problem \eqref{problem:general}, we first rewrite the constraints in \eqref{con1:general} and \eqref{con2:general}  in the form as \eqref{constraint:linear_multioperator}, with $\bH_g$ and $\bD_g$ defined as their multi-antenna counterparts,  i.e., $\bH_g=[\bH_{IT_2,g},\bH_{IT_3,g},\dots,\bH_{IT_L,g}]\in\C^{G_s\times N_T(L-1)}$ and $\mathbf{D}_g=[\bD_{IT_2,g},\bD_{IT_3,g},\dots,\bD_{IT_L,g}]\in\C^{G_s\times N_T(L-1)}$.  Assume that $\bH_g^H\bH_g=\bD_g^H\bD_g$ and are of full rank  for all $g=1,2,\dots,G$, and define  
$$\tau:=N_T(L-1)+1.$$
The solution to \eqref{problem:general} falls into two cases, depending on whether $G_s\ge \tau$ or $G_s<\tau$.

\subsubsection{$G_s\geq \tau$} 
By applying the decomposition in Proposition \ref{pro4}, the problem can be reformulated as\begin{subequations}\label{problem:general2}
\begin{align}
\max_{\bar{\bthe},\mathcal{V}\in\mathcal{X}}~&u(\bH(\bar{\bthe}),\mathcal{V})\\
\text{s.t. }~~~&\bar{\bthe}=\text{diag}(\bar{\bthe}_1,\bar{\bthe}_2,\dots,\bar{\bthe}_G), \\
&\bar{\bthe}_{g}^H\bar{\bthe}_{g}=\mathbf{I}_{G_s-\tau+1},~g=1,2,\dots, G.
\end{align}
\end{subequations}
In \eqref{problem:general2}, 
$$\bH(\bar{\bthe}):=\bar{\bH}_0+\bar{\bH}_1\bar{\bthe}\bar{\bH}_2,$$
where $\bar{\bH}_0=\sum_{g=1}^G[\bW_g]_{1:\tau-1,:}^H[\bV_g]_{1:\tau-1,:}$, $\bar{\bH}_1=[[\bW_1]_{\tau:G_s,:}^H,[\bW_2]_{\tau:G_s,:}^H,\dots,[\bW_G]_{\tau:G_s,:}^H]$, and $\bar{\bH}_2=[[\bV_1]_{\tau:G_s,:}^T,[\bV_2]_{\tau:G_s,:}^T,\dots,[\bV_G]_{\tau:G_s,:}^T]^T$ with $\bW_g=\bU(\bD_g)^H\bH_{RI,g}^H$ and $\bV_g=\bU(\bH_g)^H\bH_{IT_1,g}$.

Problem \eqref{problem:general2} can be interpreted as the utility maximization problem of a lower-dimensional RIS-aided system with fewer RIS elements, where $\bar{G}_s = G_s-\tau+1$. Therefore, existing algorithms developed for single-operator RIS-aided systems can be directly applied to solve \eqref{problem:general2}. 

\subsubsection{$G_s<\tau$} In this case, there exists a unique feasible solution to \eqref{problem:general}, given by \eqref{solution:gs<l}.
\begin{remark}
In the above discussions, we assume that $\{\bH_g\}_{1\leq g\leq G}$ are of full rank. When they are rank-deficient, which may occur, for example, when the channels associated with the non-serving operators are LoS and/or exhibit strong correlation across operators, the group-size threshold $\tau$ can be further reduced to 
\begin{equation}\label{newtau}
\tau=\max_{g\in\{1,2,\dots,G\}}\operatorname{rank}(\bH_g)+1.
\end{equation}

Specifically, the matrix $\bH_g\in\C^{G_s\times N_T(L-1)}$ can be factorized as $\bH_g=\bA_g\bB_g$, where $\bA_g\in\C^{G_s\times \operatorname{rank}(\bH_g)}$ and $\bB_g\in\C^{\operatorname{rank}(\bH_g)\times N_T(L-1)}$ are both full-rank. The constraint in \eqref{constraint:linear_multioperator} can then be rewritten as 
\begin{equation}\label{lowrank}
\bthe_g \bA_g=\bD_g\bB_g^H(\bB_g\bB_g^H)^{-1},~g=1,2,\dots,G. 
\end{equation}
Therefore, the dimension of the constraints is reduced from $N\times N_T(L-1)$ to $N\times \text{rank}(\bH_g)$ for $g=1,2,\dots,G$, and the threshold $\tau$ correspondingly reduces to \eqref{newtau}. 
\end{remark}
\begin{remark}[Generality of the proposed approach]
The proposed technique is not limited to the RIS optimization problem considered in this paper. It provides a general framework for handling optimization problems with coupled  unitary and linear equality constraints, 
 which may be useful for other structured optimization problems with similar  constraints.

\end{remark}
 \begin{figure*}
\subfigure[Rayleigh fading channels.]{\includegraphics[width=0.3\textwidth]{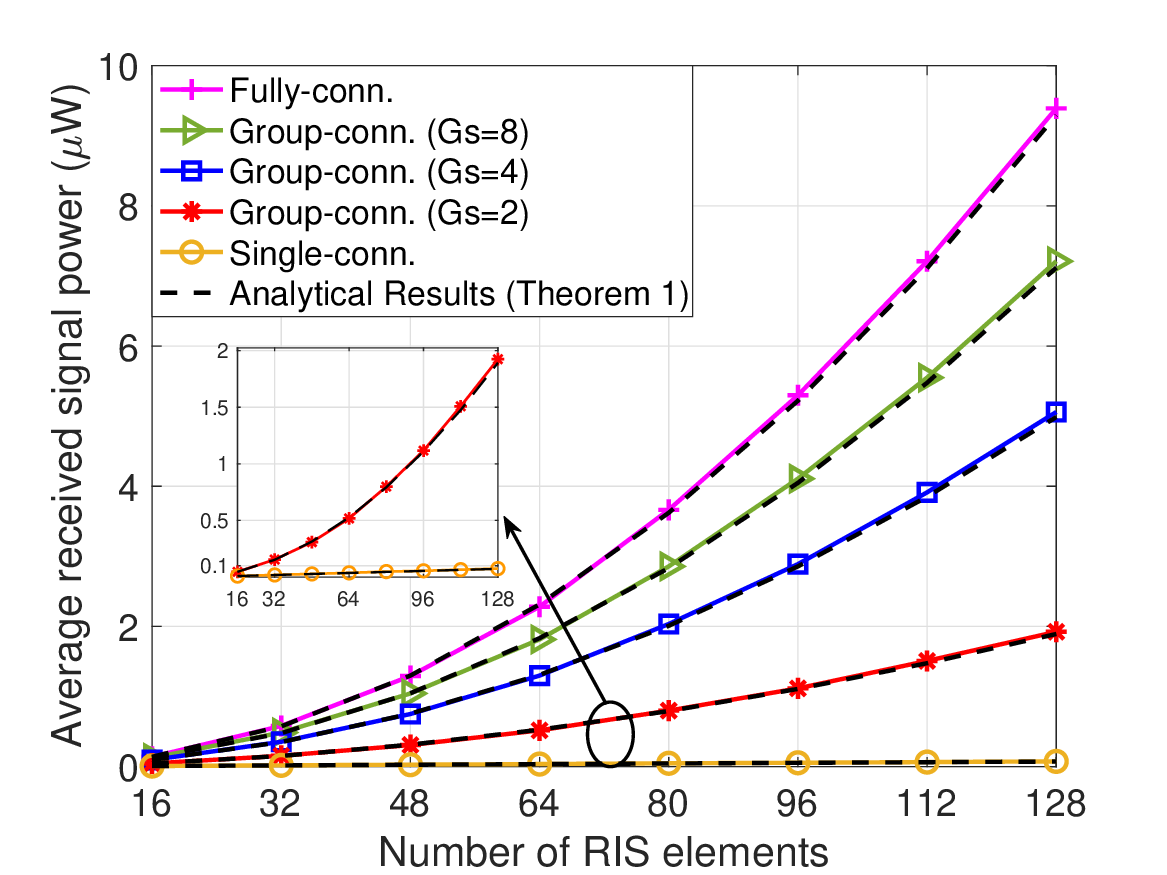}\label{operatora}}
\subfigure[LoS channels.]{\includegraphics[width=0.3\textwidth]{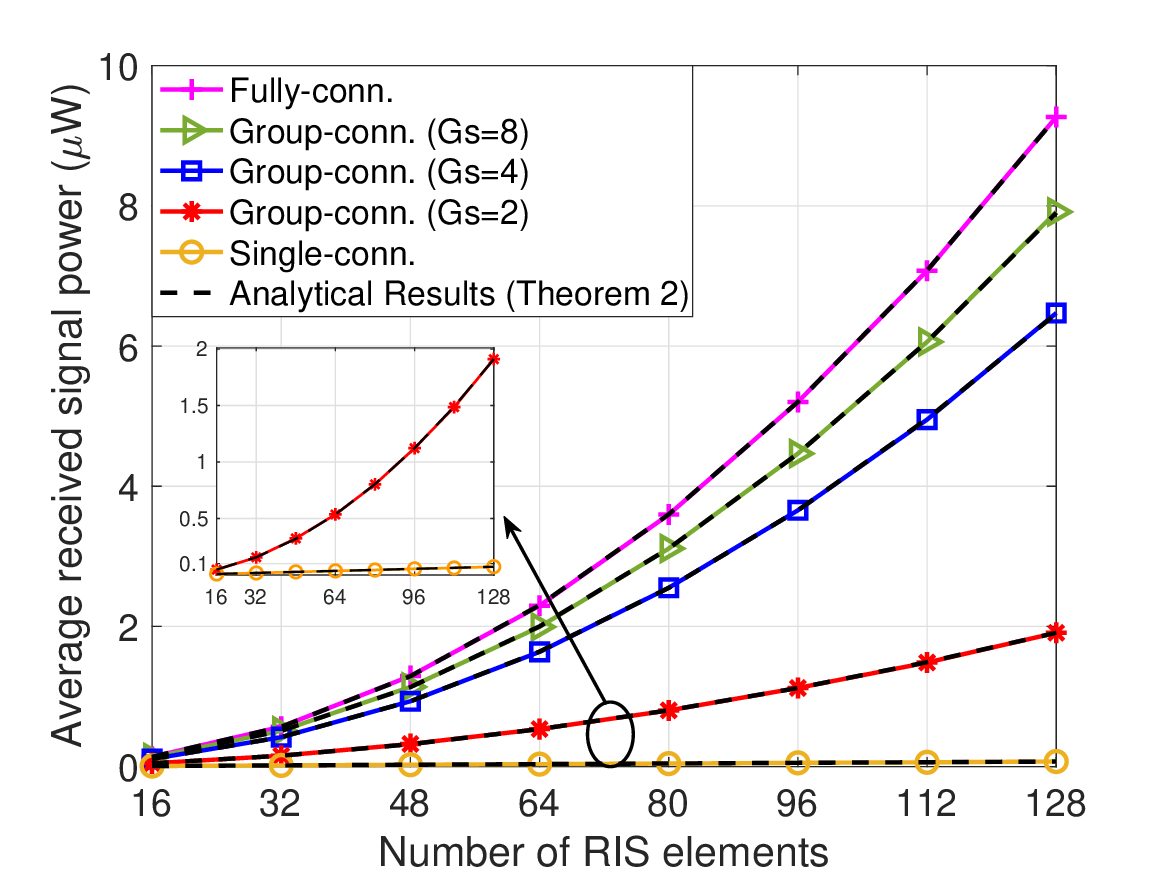}\label{operatorb}}
\subfigure[Rician fading channels.]{\includegraphics[width=0.3\textwidth]{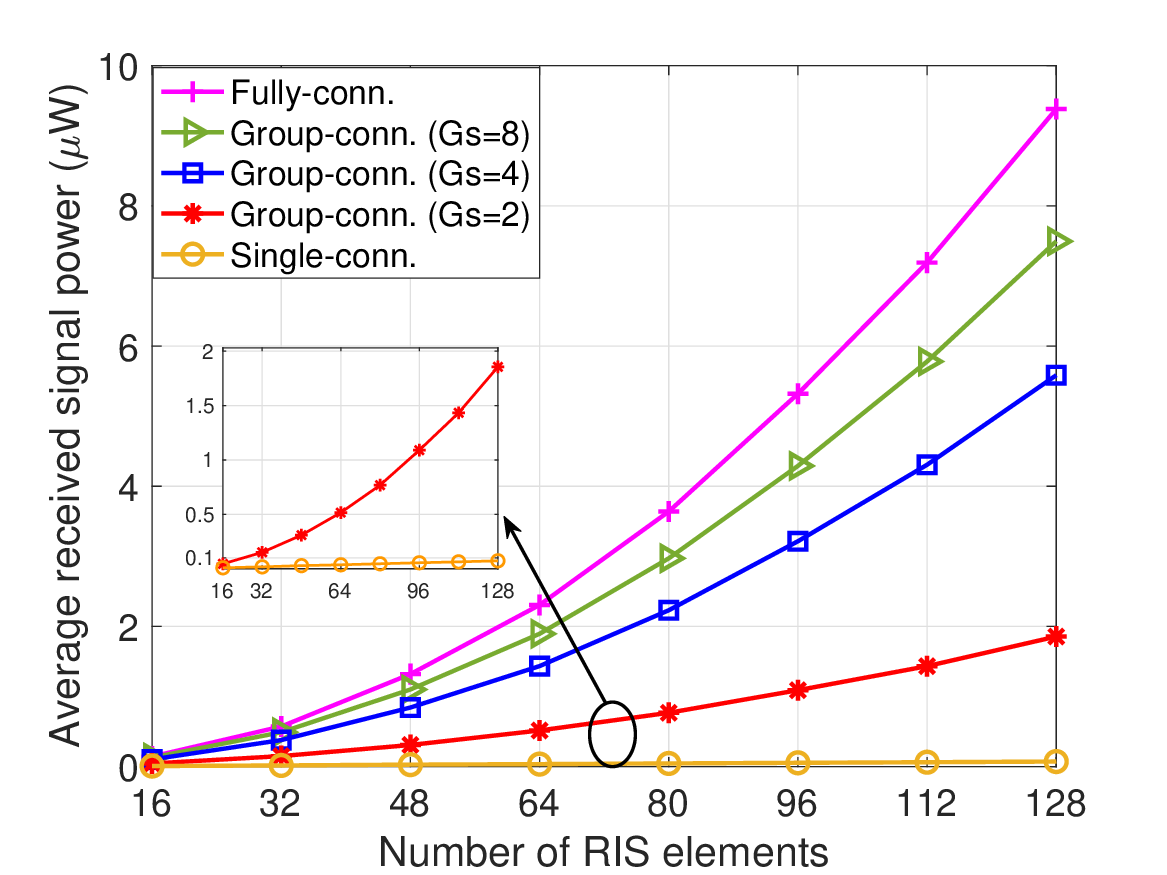}}
\centering
\caption{Average received signal power achieved by different RIS architectures versus the number of RIS elements. The number of operators is $L=2$. }
\label{fig:operator}
\end{figure*}
  \begin{figure*}
\subfigure[Rayleigh fading channels.]{\includegraphics[width=0.3\textwidth]{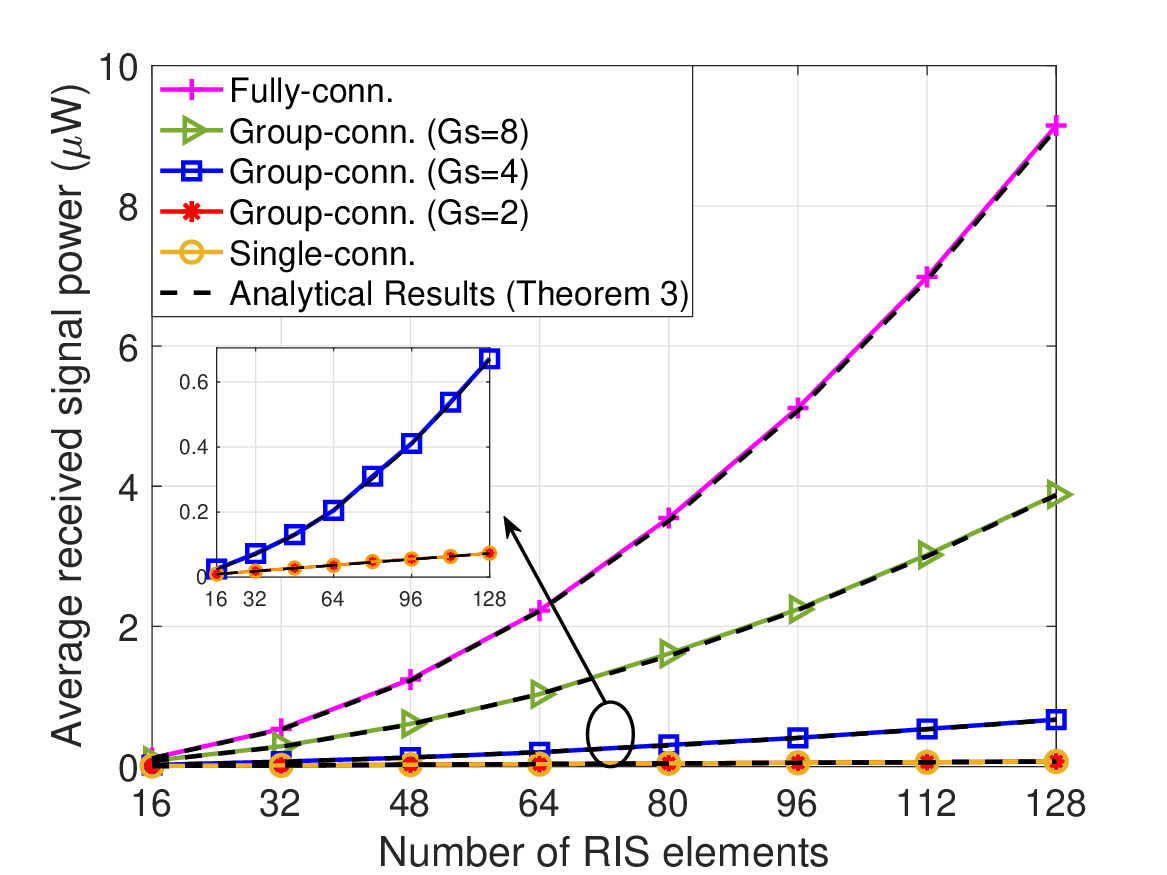}\label{operator2a}}
\subfigure[LoS channels.]{\includegraphics[width=0.3\textwidth]{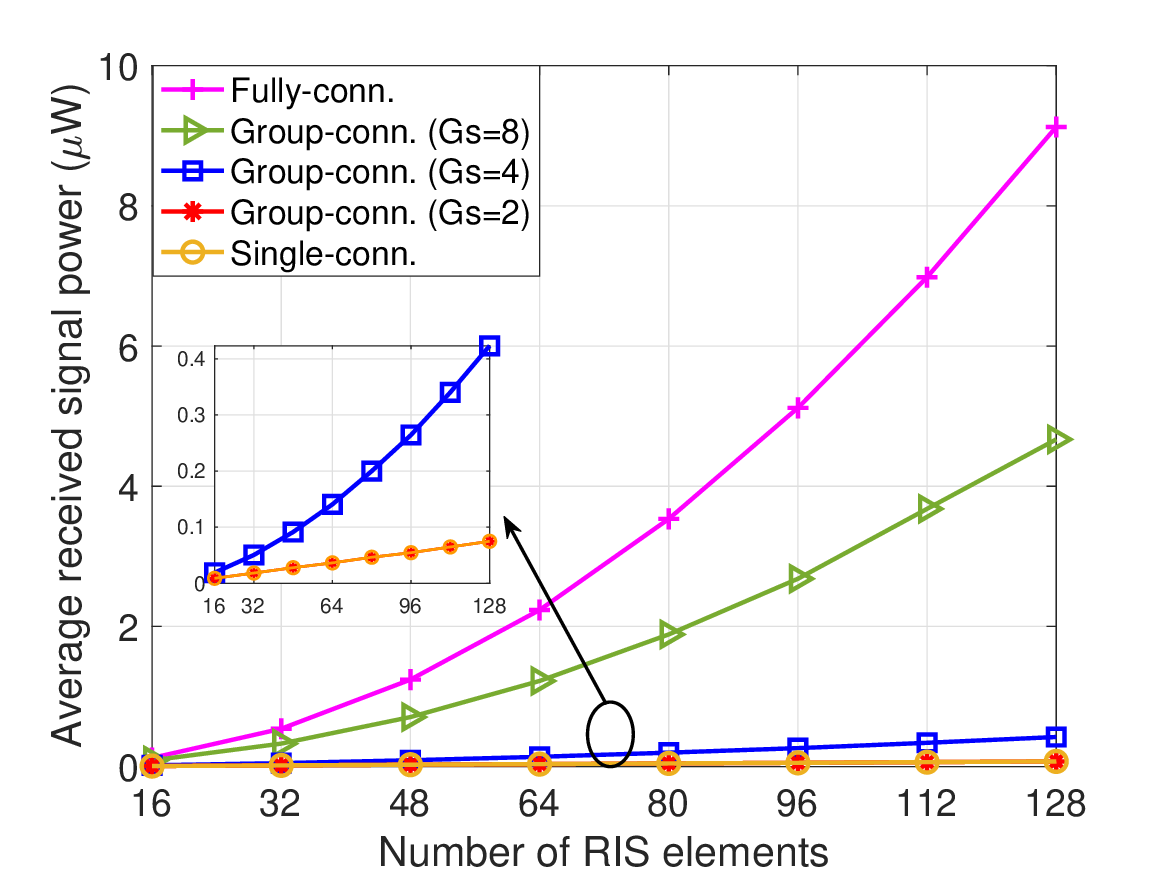}}
\subfigure[Rician fading channels.]{\includegraphics[width=0.3\textwidth]{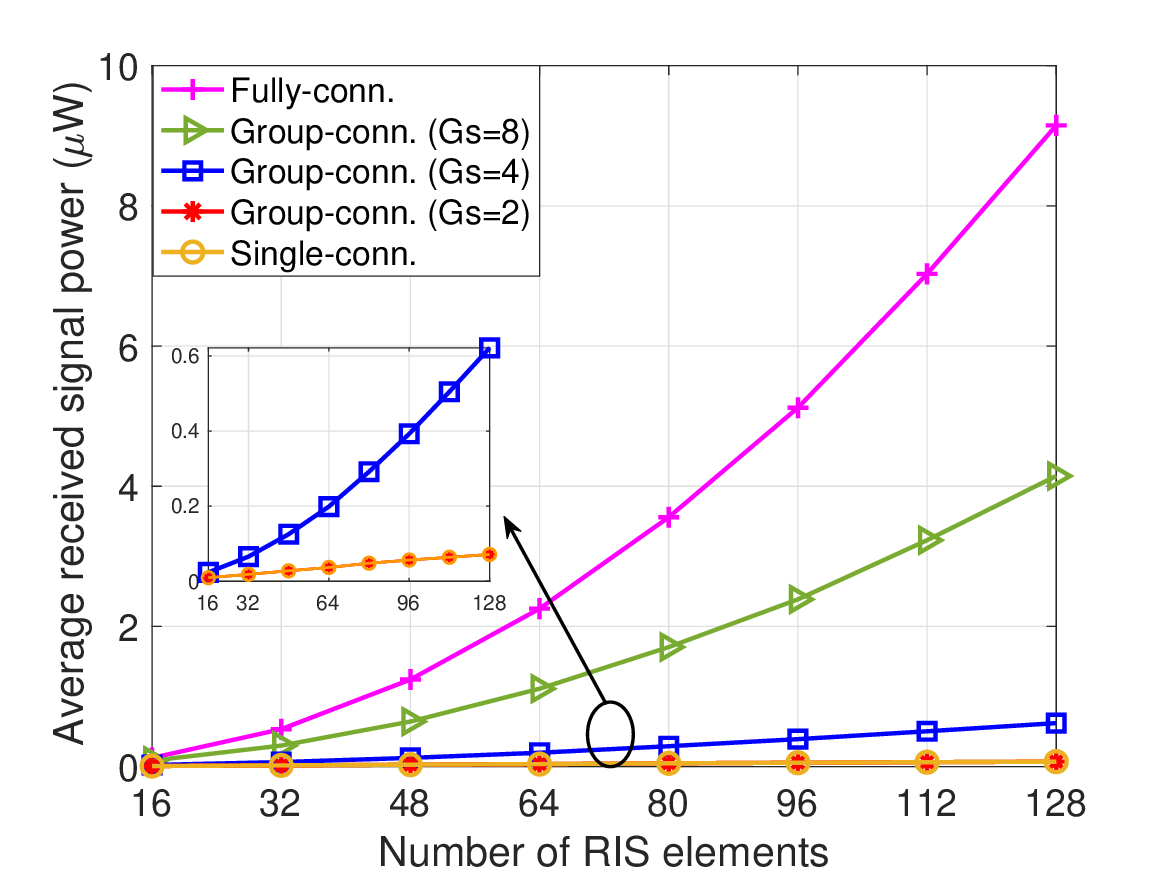}}
\centering
\caption{Average received signal power achieved by different RIS architectures versus the number of RIS elements. The number of operators is $L=4$. }
\label{fig:operator2}
\end{figure*}
Based on the scaling law analysis in Sections \ref{sec:scalinglaw} and  \ref{subsec:scalinglaw},  the key to achieving quadratic scaling is that, after satisfying the fixed RIS-reflected channel constraints imposed by the non-serving operators, the RIS still retains additional degrees of freedom to further optimize the performance of the serving operator. 
Motivated by this observation, we conjecture that a similar phenomenon also holds in multi-antenna systems, although a rigorous analysis is challenging due to the lack of a closed-form expression for the received signal power. Specifically, when $G_s \ge \tau$, the received signal power is expected to scale as $\mathcal{O}(N^2)$; otherwise, it scales as $\mathcal{O}(N)$.   The above conjecture is numerically validated  in Fig. \ref{fig:MIMO} in Section \ref{sec:simulation}. 



\section{Simulation Results}\label{sec:simulation}

In this section, we numerically evaluate the performance of different RIS architectures in multi-operator systems. Throughout this section, the distance between RIS and the BS of the serving operator, i.e., BS$_1$, is set to $d=2$m, and the distances between RIS and the BSs of the non-serving operators, i.e., BS$_2$, $\dots$, BS$_L$,  are set to $d=4$m. The distance between RIS and the user(s) served by BS$_1$ is fixed as $d=20$m. The large-scale path loss with distance $d$ is modeled as $L(d)= L_0d^{-\alpha}$. We set $L_0=-30$ dB, and adopt a path-loss exponent of $\alpha= 2.8$ for the RIS-user channel and $\alpha=2$ for the BS-RIS channels. The RIS-user channel is modeled as Rayleigh fading, while for the BS-RIS channels we consider all Rayleigh fading, LoS, and Rician fading. The Rician factor of the Rician fading channels is set to $2\,$dB. Since the BS-RIS channels vary much more slowly than the RIS-user channel, we assume that each BS-RIS channel realization remains constant for 20 RIS-user channel realizations.   The transmit power is $P_T=10$W.

In Fig. \ref{fig:operator}, we first investigate a two-operator system considered in Sections \ref{sec:solution} and \ref{sec:scalinglaw}. We depict the average received signal power for different RIS architectures of the two-operator system. The analytical results derived in Theorems \ref{scalinglaw} and \ref{scalinglaw2}, established under Rayleigh fading and LoS BS-RIS channels, are included in Figs. \ref{operatora} and \ref{operatorb}, respectively, for comparison.
  As shown in the figure, the Rayleigh, LoS, and Rician channel models exhibit the same scaling behavior. The conventional single-connected RIS performs poorly under all channel models,  with the received signal power scaling only on the order of $\mathcal{O}(N)$. In contrast, all group-connected RIS exhibit significantly better  performance.  In particular, a group-connected RIS with $G_s=2$ already achieves a receive-power scaling as $\mathcal{O}(N^2)$ in the two-operator system. As illustrated in the insets of Fig. \ref{fig:operator}, when $N=128$, the received signal power achieved by the group-connected RIS with $G_s=2$ is approximately $2\mu W$, whereas that of the single-connected RIS is around $0.1\mu W$. This corresponds to a $20$-fold (around 13 dB) gain in received signal power.  Finally, the simulated and analytical results under both Rayleigh fading and LoS channels closely match, validating the analysis in Theorems \ref{scalinglaw} and \ref{scalinglaw2}.
  
 In Fig. \ref{fig:operator2}, we further investigate the received signal power achieved by different RIS architectures in the general multi-operator system considered in Section \ref{sec:multioperator}, where the number of operators is set to $L=4$. As shown in the figure, the scaling behaviors remain unchanged across different channel models, and the analytical scaling law derived under Rayleigh fading in Theorem \ref{scalinglaw3} closely matches the simulated results. When $G_s\geq 4$, the received signal power increases with $N$ quadratically, while when $G_s<4$, it only increases linearly. 


To better visualize the impact of the number of operators on system performance, Fig. \ref{fig:wrtL} depicts the received signal power of different RIS architectures as a function of the number of operators $L$. For clarity, the theoretical curves are omitted, and only the results under Rayleigh fading BS-RIS channels are shown ( the other channel models exhibit similar trends as demonstrated in Figs. \ref{fig:operator} and \ref{fig:operator2}). As $L$ increases, all RIS architectures suffer from performance degradation, while those with larger group sizes are more robust to the increase in $L$. This validates the discussions below \eqref{wrtsingle}.   In particular, the fully-connected RIS exhibits only mild performance degradation as $L$ increases, and its received signal power remains significantly higher than that of the other group-connected RIS  when $L$ is large.  It can also be observed that a clear transition occurs when the number of operators exceeds the group size, i.e., when $L>G_s$. At the point $G_s = L-1$ highlighted in the figure, the received signal power drops to a very low level and remains almost unchanged as $L$ further increases.


  \begin{figure}
\includegraphics[scale=0.3]{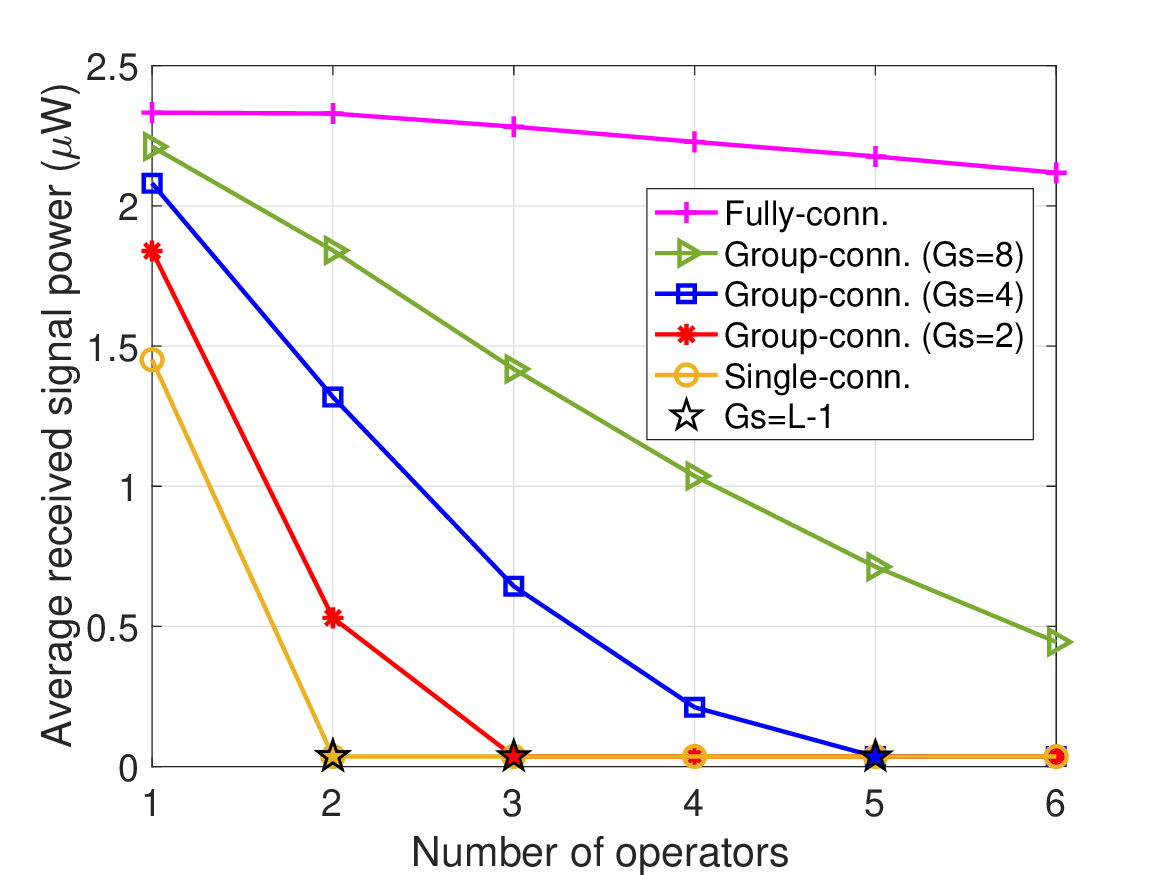}
\centering
\caption{Average received signal power of different RIS architectures versus the number of RIS elements. The number of operators is $L=4$, and the number of RIS elements is $N=64$. }
\label{fig:wrtL}
\end{figure}


In Fig. \ref{fig:MIMO}, we consider a multi-operator multi-antenna system. Specifically, the number of operators is $L=2$, and the BS of each operator is equipped with $N_T=4$ transmit antennas. The serving operator serves a single user equipped with $N_R=4$ receive antennas. Assuming single-stream transmission, we evaluate the average received signal power at the user. This corresponds to Example \ref{example1} in Section \ref{subsec:multiantenna}. We apply the algorithm in \cite{closeform} to solve the transformed problem in \eqref{problem:general2}.  As shown in the figure, the average received signal power scales quadratically with $N$ when $G_s\ge 8$, whereas it grows only linearly with $N$ when $G_s=4$ and for the single-connected RIS. This observation is consistent with our conjecture on the scaling-law behavior in Section~\ref{subsec:multiantenna}, where the theoretical transition point is given by $G_s=\tau=5$.

In Fig.~\ref{fig:multiuser}, we further consider a multi-user MISO system and evaluate its sum-rate performance, which corresponds to Example \ref{example2} in Section \ref{subsec:multiantenna}. The number of operators and the number of transmit antennas at each BS are the same as in Fig.~\ref{fig:MIMO}. The number of serving users is $N_R=4$.  The algorithm in \cite{group_conn} is applied to solve the  transformed problem in \eqref{problem:general2}. It can be observed that the sum  rate improves substantially as the group size increases from $G_s=4$ to $G_s=8$. This is because when $G_s$ exceeds the threshold $G_s=\tau=5$, the received signal power transitions from linear scaling to quadratic scaling, which leads to a significant sum-rate improvement. 
  \begin{figure}
\includegraphics[scale=0.3]{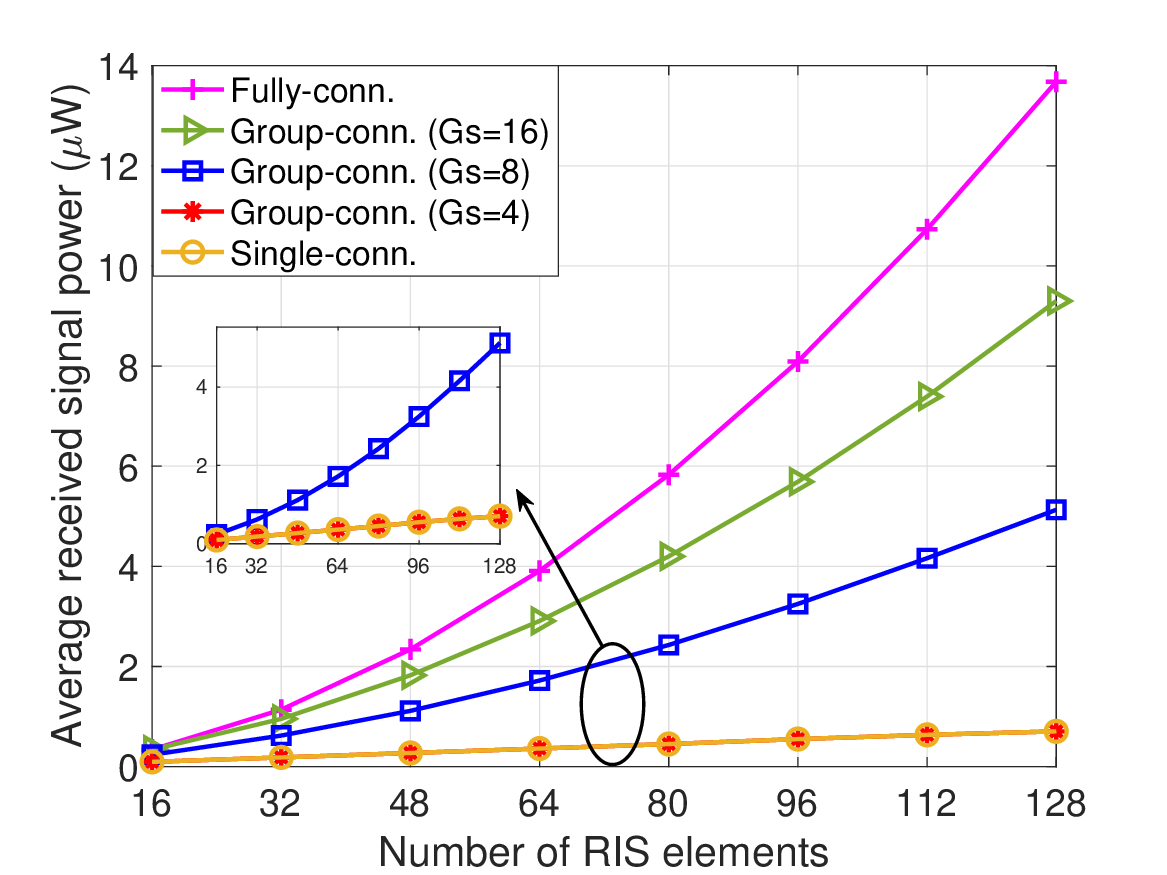}
\centering
\caption{Average received signal power of different RIS architectures versus the number of RIS elements for a single-stream MIMO system. The number of operators is $L=2$. Each BS is equipped with $N_T=4$ antennas. The number of receive antennas is $N_R=4$.}
\label{fig:MIMO}
\end{figure}
\vspace{-0.5cm}
  \begin{figure}
\includegraphics[scale=0.3]{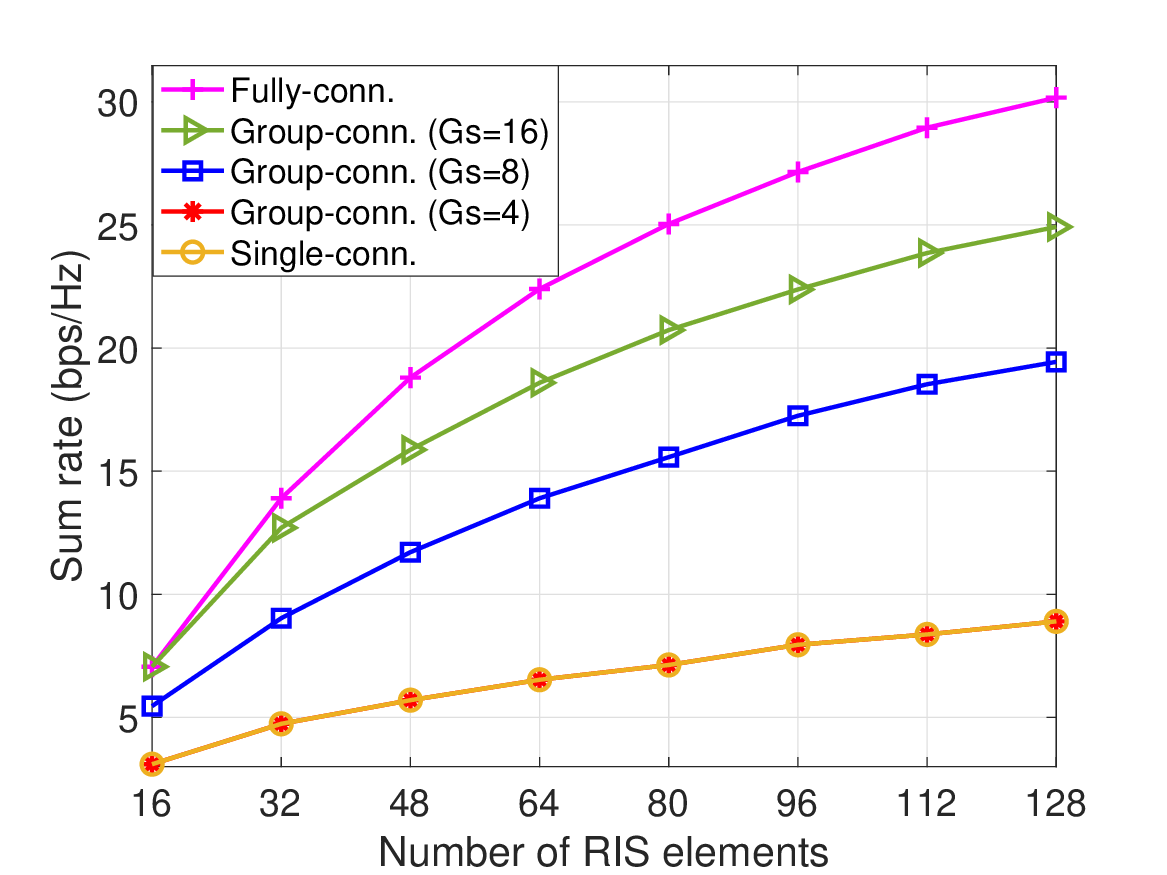}
\centering
\caption{Sum rate of different RIS architectures versus the number of RIS elements for a multi-user MISO system. The number of operators is $L=2$.  Each BS is equipped with $N_T=4$ antennas. The number of users is $N_R=4$. The noise power is set to $\sigma^2=10^{-8}$. }\label{fig:multiuser}
\end{figure}

 \section{Conclusion}\label{sec:conclusion}

This paper investigates RIS optimization and performance analysis in a multi-operator system, where one operator serves its user with the aid of RIS, while multiple non-serving operators coexist in the same environment. We formulate the problem as maximizing the received signal power of the serving operator while keeping the RIS-reflected channels of the non-serving operators fixed. To solve this problem, we develop a general framework for handling the resulting coupled linear and unitary constraints, based on which closed-form optimal solutions are derived and the corresponding received signal power  is characterized. 

Our analysis leads to the following conclusions. For the two-operator case, where each BS is equipped with a single antenna, we show that group-connected RIS with group size $G_s\geq 2$  achieves $\mathcal{O}(N^2)$ scaling in received signal power with the number of RIS elements $N$, whereas conventional single-connected RIS achieves only $\mathcal{O}(N)$ scaling. More generally, for an $L$-operator system, the scaling law of the received signal power exhibits a phase transition at the RIS group size $G_s=L$: when $G_s \geq L$, the received signal power scales as $\mathcal{O}(N^2)$, whereas when $G_s < L$,  it scales only as $\mathcal{O}(N)$. For the more general multi-antenna setting, numerical results indicate a similar phase transition at $
G_s = N_T(L-1)+1,$
where $N_T$ denotes the number of transmit antennas at each BS.

 
\appendices

\section{Proof of Theorem \ref{scalinglaw}}\label{proofscalinglaw}
Result (i) in Theorems \ref{scalinglaw}  is a special case of Result (ii) with $G=1$ and $G_s=N$. We next present the proof of Results (ii) and (iii) in Appendices \ref{proofscalinglaw:a} and \ref{proofscalinglaw:b}, respectively. 
\vspace{-0.3cm}
\subsection{Proof of Theorem \ref{scalinglaw} (ii)}\label{proofscalinglaw:a}
We first present the following lemma, which will be used frequently in the proof of Theorem \ref{scalinglaw} (ii). {\color{black} Its proof is standard and is provided in Appendix \ref{app:lemma1}. }
\begin{lemma}\label{lemma1}
Let $\x \sim \mathcal{CN}(\mathbf{0},\mathbf{I}_n)$, and let $\y \in \mathbb{C}^n$ be a unit-norm vector, i.e., $\|\y\|=1$, which is either a random vector independent of $\x$ or a deterministic vector. Then the following hold:
\begin{itemize}
    \item[(i)] $\y^H \x \sim \mathcal{CN}(0,1)$ and is independent of $\y$.
    \item[(ii)]  $\|\x\|^2-\left|{\y^H\x}\right|^2 \sim \frac{1}{2}\chi^2_{2(n-1)}$  and is independent of $\y$.
    \item[(iii)] The random variables $\y^H\x$ and $\|\x\|^2-|\y^H\x|^2$ are independent.
\end{itemize}
\end{lemma}
{

We now prove Result (ii) in Theorems \ref{scalinglaw}.
For notational simplicity, define
\begin{equation}\label{def:alpha_beta}
\begin{aligned}
\alpha_{1,g}&=\frac{\bd_g^H\h_{RI,g}}{\|\bd_g\|},~~~~\beta_{1,g}=\|\h_{RI,g}\|^2-\bigg|\frac{\bd_g^H\h_{RI,g}}{\|\bd_g\|}\bigg|^2,\\
\alpha_{2,g}&=\frac{\h_{IT_2,g}^H\h_{IT_1,g}}{\|\h_{IT_2,g}\|},~{\beta}_{2,g}=\|\h_{IT_1,g}\|^2-\bigg|\frac{\h_{IT_2,g}^H\h_{IT_1,g}}{\|\h_{IT_2,g}\|}\bigg|^2.
\end{aligned}
\end{equation}
Our goal is to  compute the expectation of the received  signal power in \eqref{group:optvalue}, which, with the above notation, can be expressed as 
\begin{equation}\label{Ef}
\begin{aligned}
\hspace{-0.1cm}\mathbb{E}[P_R^\star]=&\mathbb{E}\left[\left(\bigg|\sum_{g=1}^G\alpha_{1,g}^*\alpha_{2,g}\bigg|+\sum_{g=1}^G\beta_{1,g}^{\frac{1}{2}}\beta_{2,g}^{\frac{1}{2}}\right)^2\right]\\
=&\mathbb{E}\left[\bigg|\sum_{g=1}^G\alpha_{1,g}^*\alpha_{2,g}\bigg|^2\right]+\mathbb{E}\left[\left(\sum_{g=1}^G\beta_{1,g}^{\frac{1}{2}}\beta_{2,g}^{\frac{1}{2}}\right)^2\right]\\
&+2\mathbb{E}\left[\bigg|\sum_{g=1}^G\alpha_{1,g}^*\alpha_{2,g}\bigg|\sum_{g=1}^G\beta_{1,g}^{\frac{1}{2}}\beta_{2,g}^{\frac{1}{2}}\right]\\
:=&T_1+T_2+T_3.
\end{aligned}
\end{equation}
We next derive $T_1, T_2,$ and $T_3$ for Rayleigh fading and LoS BS-RIS channels, separately.  
According to Lemma \ref{lemma1} and the assumptions in Theorem \ref{scalinglaw},  
\begin{subequations}
\begin{align}
&\alpha_{1,g}\sim\mathcal{CN}(0,\rho_{RI}),~~{\beta}_{1,g}\sim\frac{\rho_{RI}}{2}\chi^2_{2(G_s-1)},\label{alphabeta1}\\
&\alpha_{2,g}\sim\mathcal{CN}(0,\rho_{IT_1}),~{\beta}_{2,g}\sim\frac{\rho_{IT_1}}{2}\chi^2_{2(G_s-1)}\label{alphabeta2},
\end{align}
\end{subequations}
 all of which are independent. Next, we analyze the three terms in \eqref{Ef}. First,  \vspace{-0.2cm}
\begin{equation*}
\begin{aligned}
T_1=&\sum_{g=1}^G\mathbb{E}\left[|\alpha_{1,g}|^2|\alpha_{2,g}|^2\right]\hspace{-0.05cm}+\hspace{-0.15cm}\sum_{g_1\neq g_2}\mathbb{E}[\alpha_{1,g_1}\alpha_{2,g_1}^*\alpha_{1,g_2}^*\alpha_{2,g_2}]\\
\overset{(a)}{=}&\sum_{g=1}^G\mathbb{E}\left[|\alpha_{1,g}|^2\right]\mathbb{E}\left[|\alpha_{2,g}|^2\right]=G\rho_{RI}\rho_{IT_1},
\end{aligned}
\end{equation*}
where $(a)$ holds since $\alpha_{1,g_1},\alpha_{2,g_1},\alpha_{1,g_2},\alpha_{2,g_2}$ are independent with zero mean. The second term satisfies 
\begin{equation*}
\begin{aligned}
T_2\hspace{-0.05cm}&
\overset{(a)}{=}\hspace{-0.1cm}\sum_{g=1}^G\mathbb{E}[\beta_{1,g}]\mathbb{E}[\beta_{2,g}]+\hspace{-0.2cm}\sum_{g_1\neq g_2}\hspace{-0.15cm}\mathbb{E}[\beta_{1,g_1}^{\frac{1}{2}}]\mathbb{E}[\beta_{2,g_1}^{\frac{1}{2}}]\mathbb{E}[\beta_{1,g_2}^{\frac{1}{2}}]\mathbb{E}[\beta_{2,g_2}^{\frac{1}{2}}]\\
&\overset{(b)}{=}\rho_{RI}\rho_{IT_1}\left(G(G_s-1)^2+G(G-1)\left(\frac{\Gamma(G_s-\frac{1}{2})}{\Gamma(G_s-1)}\right)^4\right),
\end{aligned}
\end{equation*}
where $(a)$ follows from the independence of the random variables, $(b)$ uses the fact that for a random variable  $X\sim\chi^2_n$, 
\begin{equation}\label{chiproperty}
\mathbb{E}[X]=n \text{ and }\mathbb{E}[X^{\frac{1}{2}}]=\sqrt{2}\frac{\Gamma(\frac{n+1}{2})}{\Gamma(\frac{n}{2})}.
\end{equation} Finally, \vspace{-0.2cm}
\begin{equation}\label{eq:T3}
\begin{aligned}
\hspace{-0.2cm}T_3&\hspace{-0.05cm}=2\mathbb{E}\left[\bigg|\sum_{g=1}^G\alpha_{1,g}^*\alpha_{2,g}\bigg|\right]\left(\sum_{g=1}^G\mathbb{E}\left[\beta_{1,g}^{\frac{1}{2}}\beta_{2,g}^{\frac{1}{2}}\right]\right)\\
&\hspace{-0.05cm}=2\hspace{-0.05cm}\sqrt{\rho_{RI}\rho_{IT_1}}G\left(\frac{\Gamma(G_s-\frac{1}{2})}{\Gamma(G_s-1)}\right)^2\hspace{-0.1cm}\mathbb{E}\hspace{-0.05cm}\left[\bigg|\sum_{g=1}^G\alpha_{1,g}^*\alpha_{2,g}\bigg|\right].
\end{aligned}
\end{equation}
Let $\boldsymbol{\alpha}_1=[\alpha_{1,1},\alpha_{1,2},\dots,\alpha_{1,G}]^T\sim\mathcal{CN}(\mathbf{0},\rho_{RI}\mathbf{I}_G)$ and $\boldsymbol{\alpha}_2=[\alpha_{2,1},\alpha_{2,2},\dots,\alpha_{2,G}]^T\sim\mathcal{CN}(\mathbf{0},\rho_{IT_1}\mathbf{I}_G)$, which are independent. Then the expectation involved in $T_3$ satisfies 
\begin{equation}\label{def:alpha_1alpha_2}
\begin{aligned}
\mathbb{E}\left[\bigg|\sum_{g=1}^G\alpha_{1,g}^*\alpha_{2,g}\bigg|\right]&=\mathbb{E}[|\boldsymbol{\alpha}_1^H\boldsymbol{\alpha}_2|]\\&\overset{(a)}{=}\mathbb{E}_{\boldsymbol{\alpha}_1}\mathbb{E}_{\boldsymbol{\alpha}_2}[|\boldsymbol{\alpha}_1^H\boldsymbol{\alpha}_2|\mid \boldsymbol{\alpha}_1]\\
&\overset{(b)}{=}\frac{\sqrt{\pi}}{2}\sqrt{\rho_{IT_1}}\mathbb{E}_{\boldsymbol{\alpha}_1}[\|\boldsymbol{\alpha}_1\|]\\&
\overset{(c)}{=}\frac{\sqrt{\pi}}{2}\sqrt{\rho_{RI}\rho_{IT_1}}\frac{\Gamma(G+\frac{1}{2})}{\Gamma(G)},
\end{aligned}
\end{equation}
where $(a)$ applies the law of total expectation, $(b)$ holds since $\boldsymbol{\alpha}_1^H\boldsymbol{\alpha}_2\sim\mathcal{CN}(\mathbf{0},\rho_{IT_1}\|\boldsymbol{\alpha}_1\|^2)$ conditioned on $\boldsymbol{\alpha}_1$ and $\mathbb{E}[|X|]=\frac{\sqrt{\pi\rho}}{2}$ for $X\sim\mathcal{CN}(0,\rho)$, and $(c)$ follows from the fact that $\|\boldsymbol{\alpha}_1\|^2\sim\frac{\rho_{RI}}{2}\chi^2_{2G}$ and \eqref{chiproperty}. Substituting \eqref{def:alpha_1alpha_2} into \eqref{eq:T3} and combining the resulting expression with $T_1$ and $T_2$ yields the desired result in Theorem \ref{scalinglaw} (ii).

\vspace{-0.2cm}
\subsection{Proof of Theorem \ref{scalinglaw} (iii)}\label{proofscalinglaw:b}
 Note that with $\bthe$ given by \eqref{theta:single}, $\bar{\h}_{RI}:=\bthe\h_{RI}\sim\mathcal{CN}(\mathbf{0},\rho_{RI}\mathbf{I}_N)$ and is independent of $\h_{IT_1}$, since $\h_{RI}$ is independent of $\h_{IT_1}$ and $\h_{IT_2}$. Under the assumption in Theorem \ref{scalinglaw}, 
\begin{equation}\label{proof:single}
\begin{aligned}
\mathbb{E}[|\bar{\h}_{RI}^H\h_{IT_1}|^2]&\overset{(a)}{=}\mathbb{E}_{\h_{IT_1}}\left[\mathbb{E}_{\h_{RI}}[\h_{IT_1}^H\bar{\h}_{RI}\bar{\h}_{RI}^H\h_{IT_1}\mid \h_{IT_1}]\right]\\&=\rho_{RI}\mathbb{E}_{\h_{IT_1}}[\|\h_{IT_1}\|^2]=N\rho_{RI}\rho_{IT_1},
\end{aligned}
\end{equation}
where $(a)$ applies the total law of expectation.  This proves Theorem \ref{scalinglaw} (iii).

\section{Proof of Lemma 1}\label{app:lemma1}
We focus on the case where $\y$ is random and independent of $\x$,  the deterministic case follows as a special case.

Since $\|\y\|=1$, its unitary completion can be written as 
$\bU(\y)=\begin{bmatrix}\y & \Y\end{bmatrix},$
where $\Y^H\Y=I_{n-1}$.  Since $\x\sim\mathcal{CN}(\mathbf{0},\mathbf{I}_n)$ and $\bU(\y)$ is unitary, conditioned on $\y$ we have
\[
\bU(\y)^H\x=
\begin{bmatrix}
\y^H\x\\
\Y^H\x
\end{bmatrix}
\sim \mathcal{CN}(\mathbf{0},\mathbf{I}_n).
\]
Hence,
$
\y^H\x\mid \y \sim \mathcal{CN}(0,1),~
\Y^H\x\mid \y \sim \mathcal{CN}(\mathbf{0},\mathbf{I}_{n-1}),
$
and these two random variables are conditionally independent. Since the above conditional distributions do not depend on $\y$, it follows that $
\y^H\x\sim\mathcal{CN}(0,1)$ and $\Y^H\x\sim \mathcal{CN}(\mathbf{0},\mathbf{I}_{n-1})$, and both are independent of $\y$, which proves (i). In addition, $
\y^H\x$ and $\bY^H\x$ are unconditionally independent. 

By the unitarity of $\bU(\y)$, the random variable in (ii) can be expressed as 
\[
\|\bU(\bY)^H\x\|^2-|\y^H\x|^2=\|\Y^H\x\|^2.
\]
Since $\Y^H\x\sim\mathcal{CN}(\mathbf{0},\mathbf{I}_{n-1})$ and is independent of $\by$, we have
$
\|\bY^H\x\|^2\sim
\frac{1}{2}\chi^2_{2(n-1)}
$
and is independent of $\y$, which proves (ii).
  Finally, (iii) follows immediately from the independence of  $\y^H\x$ and $\Y^H\x$. 
\section{Proof of Theorem \ref{scalinglaw2}}
This appendix gives the proof of Theorem \ref{scalinglaw2}. The proof follows the same procedure as that of Theorem \ref{scalinglaw}. In Appendix \ref{Theorem2ii}, we prove Theorem \ref{scalinglaw2} (ii), which includes Theorem \ref{scalinglaw2} (i) as a special case. The proof of Theorem \ref{scalinglaw2} (iii) is given in Appendix \ref{Theorem2iii}. 

\subsection{Proof of Theorem \ref{scalinglaw2} \normalfont{(ii)}} \label{Theorem2ii}
Under the assumption in Theorem \ref{scalinglaw2}, the distributions of $\alpha_{1,g}$ and $\beta_{1,g}$ are the same as those given by \eqref{alphabeta1}. In contrast, $\alpha_{2,g}$ and $\beta_{2,g}$ are deterministic and are functions of $\theta_{IT_l},~l\in\{1,2\}$. We next compute these constants.

According to the notation in Theorem \ref{scalinglaw2}, $$\h_{IT_l,g}=\sqrt{\rho_{IT_l}}[e^{-\mathrm{i}(g-1)G_s\mu_l},\dots,e^{-\mathrm{i}(gG_s-1)\mu_l}]^T,$$ and $\alpha_{2,g}$ defined in (30) can be expressed as 
$$
\begin{aligned}
\alpha_{2,g}&\overset{(a)}{=}\sqrt{\rho_{IT_1}}\frac{e^{\mathrm{i}(g-1)G_s\Delta\mu}}{\sqrt{G_s}}\sum_{n=0}^{G_s-1}e^{\mathrm{i}n\Delta\mu}\\
&\overset{(b)}{=}\sqrt{\rho_{IT_1}}\frac{e^{\mathrm{i}(g-1)G_s\Delta\mu}}{\sqrt{G_s}}e^{\mathrm{i}\frac{(G_s-1)\Delta\mu}{2}}\frac{\sin\frac{G_s\Delta\mu}{2}}{\sin\frac{\Delta\mu}{2}},
\end{aligned}$$
where $(a)$ holds since $\|\h_{IT_2,g}\|=\sqrt{\rho_{IT_2}G_s}$ and $(b)$ follows from the Dirichlet kernel identity \cite{stein2003fourier}. It follows that $|\alpha_{2,g}|^2=\rho_{IT_1}L_{G_s,\Delta\mu},$ where  the definition of $L_{G_s,\Delta\mu}$ is given in \eqref{def:L_G}.
Applying $\|\h_{IT_1,g}\|=\sqrt{\rho_{IT_1}G_s}$,  $\beta_{2,g}$ defined in \eqref{def:alpha_beta} can be written as 
$$\beta_{2,g}=\rho_{IT_1}(G_s-L_{G_s,\Delta\mu}).$$

Now we are ready to compute $T_1,T_2$ and $T_3$ in \eqref{Ef}. For $T_1$, since $\alpha_{1,g_1}$ and $\alpha_{1,g_2}$ are independent for all $g_1\neq g_2$, we have
$$
\begin{aligned}
T_1=\sum_{g=1}^G|\alpha_{2,g}|^2\mathbb{E}\left[|\alpha_{1,g}|^2\right]=GL_{G_s,\Delta\mu}\rho_{RI}\rho_{IT_1}.
\end{aligned}
$$
For $T_2$, the following result holds:
$$
\begin{aligned}
T_2\hspace{-0.05cm}&
=\hspace{-0.05cm}\sum_{g=1}^G\mathbb{E}[\beta_{1,g}]\beta_{2,g}+\hspace{-0.2cm}\sum_{g_1\neq g_2}\hspace{-0.1cm}\mathbb{E}[\beta_{1,g_1}^{\frac{1}{2}}]\mathbb{E}[\beta_{1,g_2}^{\frac{1}{2}}]\beta_{2,g_1}^{\frac{1}{2}}\beta_{2,g_2}^{\frac{1}{2}}\\
&=\rho_{RI}\rho_{IT_1}\bigg(G(G_s-1)\left(G_s-L_{G_s,\Delta\mu}\right)\\
&+G(G-1)\left(\frac{\Gamma(G_s-\frac{1}{2})}{\Gamma(G_s-1)}\right)^2(G_s-L_{G_s,\Delta\mu})\bigg).
\end{aligned}$$
Finally, following the same procedures and notations as in \eqref{eq:T3} and \eqref{def:alpha_1alpha_2}, $T_3$ satisfies
\begin{equation*}
\begin{aligned}
T_3&\hspace{-0.05cm}=2\mathbb{E}\left[\bigg|\sum_{g=1}^G\alpha_{1,g}^*\alpha_{2,g}\bigg|\right]\left(\sum_{g=1}^G\mathbb{E}\left[\beta_{1,g}^{\frac{1}{2}}\right]\beta_{2,g}^{\frac{1}{2}}\right)\\
&\hspace{-0.05cm}=2\hspace{-0.05cm}\sqrt{\rho_{RI}\rho_{IT_1}}G\frac{\Gamma(G_s-\frac{1}{2})}{\Gamma(G_s-1)}\left(G_s-L_{G_s,\Delta\mu}\right)^{\frac{1}{2}} \mathbb{E}\left[|\boldsymbol{\alpha}_1^H\boldsymbol{\alpha}_2|\right],
\end{aligned}
\end{equation*}
where $\boldsymbol{\alpha}_1^H\boldsymbol{\alpha}_2\sim\mathcal{CN}(0,\rho_{RI}\|\boldsymbol{\alpha}_2\|^2)$, and thus  
$$\begin{aligned}
\mathbb{E}\left[|\boldsymbol{\alpha}_1^H\boldsymbol{\alpha}_2|\right]&=\frac{\sqrt{\pi\rho_{RI}}}{2}\|\boldsymbol{\alpha}_2\|=\frac{\sqrt{\pi\rho_{RI}\rho_{IT_1} G}}{2}L_{G_s,\Delta\mu}.
\end{aligned}
$$
Combining the above yields Theorem \ref{scalinglaw2} (ii). 
\subsection{Proof of Theorem \ref{scalinglaw2} \normalfont{(iii)}} \label{Theorem2iii}
Under the assumption in Theorem \ref{scalinglaw2}, $\bar{\h}_{RI}^H\h_{IT_1}\sim\mathcal{CN}(0,N\rho_{RI}\rho_{IT_1})$ due to the fact that $\|\h_{IT_1}\|^2=N\rho_{IT_1}$. Hence,
$$
\begin{aligned}
\mathbb{E}[|\bar{\h}_{RI}^H\h_{IT_1}|^2]=N\rho_{RI}\rho_{IT_1}.
\end{aligned}$$
This proves Theorem \ref{scalinglaw2} (iii).
\section{Proof of Proposition \ref{pro4}}
First, given a matrix $\bthe_g$ satisfying constraints \eqref{con03} and \eqref{constraint:linear_multioperator}, define 
\begin{equation}\label{def:Q_2}
\mathbf{Q}_g:=\bU(\mathbf{D}_g)^H\bthe_g\bU(\bH_g), \tag{37}
\end{equation}
where $\bU(\mathbf{D}_g)$ and $\bU(\mathbf{H}_g)$ are any unitary completion of $\bD_g$ and $\bH_g$, respectively. Clearly, $\bQ_g$ is unitary  since $\bU(\mathbf{D}_g)$, $\bthe_g$, and $\bU(\bH_g)$ are all unitary. In addition, the first $L-1$ columns of $\mathbf{Q}_g$ satisfy 
$$
\begin{aligned}
\mathbf{Q}_g\mathbf{E}_{L-1}&=\bU(\mathbf{D}_g)^H\bthe_g\bU(\bH_g)\mathbf{E}_{L-1}\\
&\overset{(a)}{=}\bU(\mathbf{D}_g)^H\bthe_g\bH_g(\bH_g^H\bH_g)^{-\frac{1}{2}}\\
&\overset{(b)}{=}\bU(\mathbf{D}_g)^H\mathbf{D}_g(\bH_g^H\bH_g)^{-\frac{1}{2}}\\
&\overset{(c)}{=}\mathbf{E}_{L-1}(\bD_g^H\bD_g)^{\frac{1}{2}}(\bH_g^H\bH_g)^{-\frac{1}{2}}\\
&\overset{(d)}{=}\mathbf{E}_{L-1},
\end{aligned}
$$
where (a) follows from the definition of the unitary completion in Definition \ref{def:UC_matrix}, i.e., $\bU(\bH_g)\mathbf{E}_{L-1}= \bH_g(\bH_g^H\bH_g)^{-\frac{1}{2}}$,  (b) applies \eqref{constraint:linear_multioperator},  (c) holds since $\bU(\mathbf{D}_g)^H\mathbf{D}_g=\mathbf{E}_{L-1}(\bD_g^H\bD_g)^{\frac{1}{2}}$, which is obtained by multiplying both sides of $\bU(\mathbf{D}_g)\mathbf{E}_{L-1}=\mathbf{D}_g(\bD_g^H\bD_g)^{-\frac{1}{2}}$ by $\bU(\bD_g)^H$, and (d) uses  $\bH_g^H\bH_g=\bD_g^H\bD_g$. Therefore, $\bQ_g$ can be expressed as $$\bQ_g=\left[\begin{matrix}\mathbf{I}_{L-1}&{\bB_g}\\\mathbf{0}&\bar{\bQ}_g\end{matrix}\right].$$ Using $\bQ_g\bQ_g^H=\mathbf{I}_{G_s}$, we get  $\bar{\bQ}_g\bar{\bQ}_g^H=\mathbf{I}_{G_s-L+1}$ and 
$$\mathbf{I}_{L-1}+\bB_g^H\bB_g=\mathbf{I}_{L-1},$$ which further implies that  ${\bB}_g=\mathbf{0}$. This, combined with \eqref{def:Q_2}, indicates that  $\bthe_g$ can be written as in Proposition \ref{pro4} with $\bar{\bthe}_g=\bar{\bQ}_g$.

On the other hand, let ${\bthe}_g$ be a matrix in the form of the decomposition in Proposition \ref{pro4}. It is easy to check that $\bthe_g$ is unitary, i.e., $\bthe_g$ satisfies \eqref{con03}. In addition,  
$$
\begin{aligned}
\bthe_g\bH_g&=\bU(\mathbf{D}_g)\left[\begin{matrix}\mathbf{I}_{L-1}&\mathbf{0}\\\mathbf{0}&\bar{\bthe}_g\end{matrix}\right]\bU(\bH_g)^H\bH_g\\
&\overset{(a)}{=}\bU(\bD_g)\left[\begin{matrix}\mathbf{I}_{L-1}&\mathbf{0}\\\mathbf{0}&\bar{\bthe}_g\end{matrix}\right]\mathbf{E}_{L-1}(\bH_g^H\bH_g)^{\frac{1}{2}}\\
&=\bU(\mathbf{D}_g)\mathbf{E}_{L-1}(\bH_g^H\bH_g)^{\frac{1}{2}}\\
&\overset{(b)}{=}\bD_g(\mathbf{D}_g^H\mathbf{D}_g)^{-\frac{1}{2}}(\bH_g^H\bH_g)^{\frac{1}{2}}
\\&\overset{(c)}{=}\mathbf{D}_g,
\end{aligned}
$$
i.e., \eqref{constraint:linear_multioperator} is satisfied, where (a) is due to $\bU(\bH_g)^H\bH_g= \mathbf{E}_{L-1}(\bH_g^H\bH_g)^{\frac{1}{2}}$, (b) uses the definition that $\bU(\mathbf{D}_g)\mathbf{E}_{L-1}=\bD_g(\mathbf{D}_g^H\mathbf{D}_g)^{-\frac{1}{2}}$, and (c) follows from $\bH_g^H\bH_g=\bD_g^H\bD_g$.  This completes the proof. 

\section{Proof of Theorem 3}
{
For the case that $G_s\geq L$, define 
\begin{equation*}\label{def:alpha_beta_2}
\begin{aligned}
\boldsymbol{\alpha}_{1,g}&=(\bD_g^H\bD_g)^{-\frac{1}{2}}\bD_g^H\h_{RI,g},~\beta_{1,g}=\|\h_{RI,g}\|^2-\|\boldsymbol{\alpha}_{1,g}\|^2,\\
\boldsymbol{\alpha}_{2,g}&=(\bH_g^H\bH_g)^{-\frac{1}{2}}\bH_g^H\h_{IT_1,g},~{\beta}_{2,g}=\|\h_{IT_1,g}\|^2-\|\boldsymbol{\alpha}_{2,g}\|^2.
\end{aligned}
\end{equation*}
Analogous to Lemma \ref{lemma1},  we can prove that the above random variables satisfy
 \begin{equation}\label{multioperator:alpha}\tag{38}
\begin{aligned}
&\boldsymbol{\alpha}_{1,g}\sim\mathcal{CN}(0,\rho_{RI}\mathbf{I}_{L-1}),~~{\beta}_{1,g}\sim\frac{\rho_{RI}}{2}\chi^2_{2(G_s-L+1)},\\
&\boldsymbol{\alpha}_{2,g}\sim\mathcal{CN}(0,\rho_{IT_1}\mathbf{I}_{L-1}),~{\beta}_{2,g}\sim\frac{\rho_{IT_1}}{2}\chi^2_{2(G_s-L+1)}, 
\end{aligned}
\end{equation}
all of which are independent.  According to \eqref{group:multioperator}, $\mathbb{E}[P_R^\star]$ can be expressed as 
\begin{equation*}
\begin{aligned}
\hspace{-0.1cm}\mathbb{E}[P_R^\star]=&\mathbb{E}\left[\left(\bigg|\sum_{g=1}^G\boldsymbol{\alpha}_{1,g}^H\boldsymbol{\alpha}_{2,g}\bigg|+\sum_{g=1}^G\beta_{1,g}^{\frac{1}{2}}\beta_{2,g}^{\frac{1}{2}}\right)^2\right]\\
=&\mathbb{E}\left[\bigg|\sum_{g=1}^G\boldsymbol{\alpha}_{1,g}^H\boldsymbol{\alpha}_{2,g}\bigg|^2\right]+\mathbb{E}\left[\left(\sum_{g=1}^G\beta_{1,g}^{\frac{1}{2}}\beta_{2,g}^{\frac{1}{2}}\right)^2\right]\\
&+2\mathbb{E}\left[\bigg|\sum_{g=1}^G\boldsymbol{\alpha}_{1,g}^H\boldsymbol{\alpha}_{2,g}\bigg|\sum_{g=1}^G\beta_{1,g}^{\frac{1}{2}}\beta_{2,g}^{\frac{1}{2}}\right]\\
:=&T_1+T_2+T_3.
\end{aligned}
\end{equation*}
Using the distributions in \eqref{multioperator:alpha} and following techniques analogous to those in Appendix  A, the three terms can be expressed and computed  as follows.
The first term is  
$$
\begin{aligned}
T_1&=\sum_{g=1}^G\mathbb{E}\left[|\boldsymbol{\alpha}_{1,g}^H\boldsymbol{\alpha}_{2,g}|^2\right]\\
&=\sum_{g=1}^G\mathbb{E}_{\boldsymbol{\alpha}_{1,g}}\mathbb{E}[|\boldsymbol{\alpha}_{1,g}^H\boldsymbol{\alpha}_{2,g}|^2\mid\boldsymbol{\alpha}_{1,g}]\\
&=\rho_{IT_1}\sum_{g=1}^G\mathbb{E}_{\boldsymbol{\alpha}_{1,g}}[\|\boldsymbol{\alpha}_{1,g}\|^2]=G\rho_{RI}
\rho_{IT_1}(L-1).
\end{aligned}
$$
The second term satisfies 
$$
\begin{aligned}
&T_2\hspace{-0.02cm}=\hspace{-0.05cm}\sum_{g=1}^G\mathbb{E}[\beta_{1,g}]\mathbb{E}[\beta_{2,g}]+\hspace{-0.2cm}\sum_{g_1\neq g_2}\hspace{-0.1cm}\mathbb{E}[\beta_{1,g_1}^{\frac{1}{2}}]\mathbb{E}[\beta_{2,g_1}^{\frac{1}{2}}]\mathbb{E}[\beta_{1,g_2}^{\frac{1}{2}}]\mathbb{E}[\beta_{2,g_2}^{\frac{1}{2}}]\\
&=\rho_{RI}\rho_{IT_1}\hspace{-0.1cm}\left(\hspace{-0.05cm}G(G_s-L+1)^2\hspace{-0.05cm}+G(G-1)\left(\frac{\Gamma(G_s-L+\frac{3}{2})}{\Gamma(G_s-L+1)}\right)^{\hspace{-0.05cm}4}\right),
\end{aligned}$$
where we apply \eqref{multioperator:alpha} and the properties of the chi-square distribution in \eqref{chiproperty}.
The last term is given by 
\begin{equation*}
\begin{aligned}
T_3&\hspace{-0.05cm}=2\mathbb{E}\left[\bigg|\sum_{g=1}^G\boldsymbol{\alpha}_{1,g}^H\boldsymbol{\alpha}_{2,g}\bigg|\right]\left(\sum_{g=1}^G\mathbb{E}\left[\beta_{1,g}^{\frac{1}{2}}\beta_{2,g}^{\frac{1}{2}}\right]\right)\\
&\hspace{-0.05cm}=2\hspace{-0.05cm}\sqrt{\rho_{RI}\rho_{IT_1}}G\left(\frac{\Gamma(G_s-L+\frac{3}{2})}{\Gamma(G_s-L+1)}\right)^2\hspace{-0.1cm}\mathbb{E}\hspace{-0.05cm}\left[\bigg|\sum_{g=1}^G\boldsymbol{\alpha}_{1,g}^H\boldsymbol{\alpha}_{2,g}\bigg|\right].
\end{aligned}
\end{equation*}
Let $\boldsymbol{\alpha}_1=[\boldsymbol{\alpha}_{1,1}^T,\dots,\boldsymbol{\alpha}_{1,G}^T]^T\sim\mathcal{CN}(\mathbf{0},\rho_{RI}\mathbf{I}_{(L-1)G})$ and $\boldsymbol{\alpha}_2=[\boldsymbol{\alpha}_{2,1}^T,\dots,\boldsymbol{\alpha}_{2,G}^T]^T\sim\mathcal{CN}(\mathbf{0},\rho_{IT_1}\mathbf{I}_{(L-1)G})$. Since $\boldsymbol{\alpha}_1$ and $\boldsymbol{\alpha}_2$ are independent, we have $\boldsymbol{\alpha}_{1}^H\boldsymbol{\alpha}_{2}\mid\boldsymbol{\alpha}_1\sim\mathcal{CN}(0,\rho_{IT_1}\|\boldsymbol{\alpha}_1\|^2)$. Similar to \eqref{def:alpha_1alpha_2}, we get
$$\begin{aligned}
\mathbb{E}\hspace{-0.05cm}\left[|\boldsymbol{\alpha}_{1}^H\boldsymbol{\alpha}_{2}|\right]&=\mathbb{E}_{\boldsymbol{\alpha}_1}\mathbb{E}_{\boldsymbol{\alpha}_2}\left[|\boldsymbol{\alpha}_{1}^H\boldsymbol{\alpha}_{2}|\mid\boldsymbol{\alpha}_1\right]\\
&=\frac{\sqrt{\pi}}{2}\sqrt{\rho_{IT_1}}\mathbb{E}_{\boldsymbol{\alpha}_1}[\|\boldsymbol{\alpha}_1\|]\\&
=\frac{\sqrt{\pi}}{2}\sqrt{\rho_{RI}\rho_{IT_1}}\frac{\Gamma(G(L-1)+\frac{1}{2})}{\Gamma(G(L-1))}.
\end{aligned}
$$
Combining the above yields Theorem \ref{scalinglaw3} (i).

For the case $G_s< L$,  we have $$\bar{\h}_{RI,g}:=(\bH_g\bH_g^H)^{-1}\bH_g\bD_g^H\h_{RI,g}\sim\mathcal{CN}(\mathbf{0},\rho_{RI}\mathbf{I}_{G_s}),$$
since $\bH_g$ and $\bD_g$ are independent of $\h_{RI,g}$, and satisfy $\bH_g^H\bH_g=\bD_g^H\bD_g$.  It follows that 
$$\bar{\h}_{RI}:=[\bar{\h}_{RI,1}^T,\bar{\h}_{RI,2}^T,\dots,\bar{\h}_{RI,G}^T]^T\sim\mathcal{CN}(\mathbf{0},\rho_{RI}\mathbf{I}_N).$$ According to \eqref{multi:Gs<L}, 
$$\mathbb{E}[P_R^\star]=\mathbb{E}\left[\left|\sum_{g=1}^G\bar{\h}_{RI,g}^H\h_{IT_1,g}\right|^2\right]=\mathbb{E}\left[|\bar{\h}_{RI}^H\h_{IT_1}|^2\right].$$
Similar to \eqref{proof:single}, we get 
$$\mathbb{E}[P_R^\star]=N\rho_{RI}\rho_{IT_1}.$$
{

  \bibliographystyle{IEEEtran}
\bibliography{IEEEabrv,BDRIS}
 \end{document}